\documentclass{amsart}[12pt]  
\usepackage{amsmath}
\usepackage{graphicx}
\usepackage{graphicx,xcolor}
\usepackage{amsfonts} 
\usepackage{amssymb}
\usepackage{pdfsync}
\usepackage{color}
\usepackage{geometry}
\usepackage{calrsfs}
\DeclareMathAlphabet{\pazocal}{OMS}{zplm}{m}{n}

\providecommand{\keywords}[1]{\textbf{\textit{Index terms---}} #1}
\usepackage{epsfig}
\theoremstyle{plain}
\newtheorem{theorem}{Theorem}[section]

\newtheorem{lemma}[theorem]{Lemma}

\newtheorem*{assumption*}{Assumption}
\newtheorem{proposition}[theorem]{Proposition}

\newtheorem{remark}[theorem]{Remark}

\theoremstyle{definition}
\newtheorem*{remark*}{Remark}

\theoremstyle{remark}
\numberwithin{equation}{section}

\newcommand{\ep}{\varepsilon}

\newcommand{\ffi}{\varphi}

\newcommand{\R}{\mathbb{R}}

\newcommand{\PP}{\mathbb{P}}
\newcommand{\EE}{\mathbb{E}}

\newcommand{\T}{\mathcal{T}}
\newcommand{\LL}{\mathcal{L}}

\newcommand{\mbf}{\mathbf}

\newcommand{\supp}{\mathrm{supp}\,}

\newcommand{\h}{h_{\ep,\gamma}}

\newcommand{\norm}[1]{\Vert #1\Vert}

\date{\today}

\title[]{Derivation of the linear Landau equation and linear Boltzmann equation from the Lorentz model with magnetic field}
\author[M. Marcozzi]{M. Marcozzi}
\address[Matteo Marcozzi]{University of Helsinki,
Department of Mathematics and Statistics
FI-00014 Helsingin yliopisto, Finland}
\email[M. Marcozzi]{matteo.marcozzi@helsinki.fi}
\author[A. Nota]
{A. Nota}
\address[Alessia Nota]{Institute for Applied Mathematics,
University of Bonn, Endenicher Allee 60,
D-53115 Bonn, Germany }
\email[A. Nota]{nota@iam.uni-bonn.de}

\begin{document}
\maketitle

\begin{abstract}
We consider a test particle moving in a random distribution of obstacles in the plane, under the action of a uniform magnetic field, orthogonal to the plane. We show that, in a weak coupling limit, the particle distribution behaves according to the linear Landau equation with a magnetic transport term. Moreover, we show that, in a low density regime, 
when each obstacle generates an inverse power law potential, the particle distribution behaves according to the linear Boltzmann equation with a magnetic transport term. We provide an explicit control of the error in the kinetic limit by estimating the contributions of the configurations which prevent the Markovianity. 
We compare these results with those ones obtained for a system of hard disks in \cite{BMHH}, 
which show instead that the memory effects are not negligible in the Boltzmann-Grad limit. 
\end{abstract}

\vspace*{2mm}

\tableofcontents

\section{Introduction}
Consider a point particle of mass $ m=1$ in $\R^d$, $d=2,3$ moving in a random distribution of fixed scatterers, whose centers are denoted by $(c_1,\dots,c_N)$.\\
We assume that the scatterers are distributed according to a Poisson distribution of parameter $\mu>0$. The equations of motion are
\begin{equation}\label{eq:motion_no_B}
\left\{\begin{array}{ll}
\dot{x}=v&\\
\dot{v}=-\sum_{i}\nabla\phi(|x-c_i|)&,
\end{array}\right.
\end{equation}
here $(x,v)$ denote position and velocity of the test particle, $t$ the time and $\dot{A}=\frac{\,dA}{\,dt}$ for any time dependent variable $A$.

To outline a kinetic behavior it is usually introduced a scaling of the space-time variables and the density of the scatterer distribution.
For this model, it is more physically intuitive to transfer the scaling to the background medium. 
More precisely, let $\ep>0$ be a parameter indicating the ratio between the macroscopic and microscopic variables, 
we keep time and space fixed and rescale the range of the interaction and the density of the scatterers, i.e.
\begin{equation}\label{eq:scaling}
\left.\begin{array}{ll}\vspace{1.5mm}
\phi_\ep(x)=\ep^{\alpha}\phi(\frac{x}{\ep})&\\
\mu_{\varepsilon}=\mu\,\varepsilon^{-(d-1+2\alpha)}
\end{array}\right.
\end{equation} 
where $d=2,3$ is the dimension of the physical space and $\alpha\in[0,\frac{1}{2}]$ is a suitable parameter.  
This means that the probability of finding $N$ obstacles in a bounded measurable set $\Lambda\subset\R^d$ is given by
\begin{equation}\label{poisson}
\PP_{\ep}(\,d\mbf{c}_{N})=e^{-\mu_{\ep}|\Lambda|}\frac{\mu_{\varepsilon}^{N}}{N!}\,dc_1\cdots\,dc_N
\end{equation}
where $\mbf{c}_{N}=c_1,\dots,c_N$ and $|\Lambda|=\text{meas}(\Lambda)$.
Consequently, the equation of motion \eqref{eq:motion_no_B} becomes
\begin{equation}\label{eq:scaled_no_B}
\left\{\begin{array}{ll}
\dot{x}=v&\\
\dot{v}=-\ep^{\alpha-1}\sum_{i}\nabla\phi(\frac{|x-c_i|}{\ep})&.
\end{array}\right.
\end{equation}
Now let $T^t_{\mbf{c}_N}(x,v)$ be the Hamiltonian flow solution to equation \eqref{eq:scaled_no_B} with initial datum $(x,v)$ in a given sample 
of obstacles (skipping the $\ep$ dependence for notational simplicity) and, for a given probability distribution $f_0=f_0(x,v)$, consider the quantity
\begin{equation}\label{expect_density}
f_{\ep}(x,v,t)=\EE_{\ep}[f_0(T^{-t}_{\mbf{c}_N}(x,v))]
\end{equation}
where $\EE_{\ep}$ is the expectation with respect to the measure $\PP_{\ep}$ given by \eqref{poisson}.

In the limit $\ep\rightarrow 0$ we expect that the probability distribution \eqref{expect_density} will solve a linear kinetic equation depending on the value of $\alpha$. 
If $\alpha=0$ the limit corresponding to such a scaling is called low-density (or Boltzmann-Grad) limit. In this case $f_{\ep}$ converges to the solution of a linear Boltzmann equation. See \cite{G,S,BBS,DP}.
On the other hand, if $\alpha=\frac{1}{2}$ the corresponding limit, called weak-coupling limit, yields the linear Landau equation, as proven in 
\cite{KP,DGL,K}.
The intermediate scaling, namely $\alpha\in(0,\frac{1}{2})$, although refers to a low-density regime, leads to the linear Landau equation again, see \cite{DR,K}.

We want to remark that in \cite{DP,DR} the authors exploit the original constructive idea due to Gallavotti (see \cite{G}) for the Boltzmann-Grad limit. This method is based on a suitable change of variables which can be implemented outside a set $ E$ of pathological events which prevent the Markov property of the limit (such as the set of configurations yielding recollisions, i.e. when the test particle recollides with a given obstacle after having
suffered collisions with other different obstacles). The probability $\PP_{\ep}(E)$ is vanishing as $ \ep$ tends to 0. The main difference is that in \cite{DP} the range of the potential is infinite in the limit, therefore the test particle interacts with infinitely many obstacles. As for the case of the long range potential considered in \cite{DP} also in \cite{DR} there is a lack of the semi explicit form of the solution of the limit equation. This requires explicit estimates for the set of bad configurations of obstacles. For a short range potential, like in the case of the hard-sphere potential considered in \cite{G}, a simple dimensional argument is sufficient. For an explicit control of the error in the kinetic limit for the hard-sphere potential see for instance \cite{BNPP}.
Moreover, in \cite{DR,K} it was proven that even if $\alpha>0$, but sufficiently small, the recollisions are still negligible. Incidentally we note that, if $\alpha$ is close to $1/2$, this is not true anymore and it would be interesting to derive the Landau equation in this regime, by means of an explicit constructive approach.

\begin{figure}[t] 
   \centering
  \includegraphics[scale=0.4]{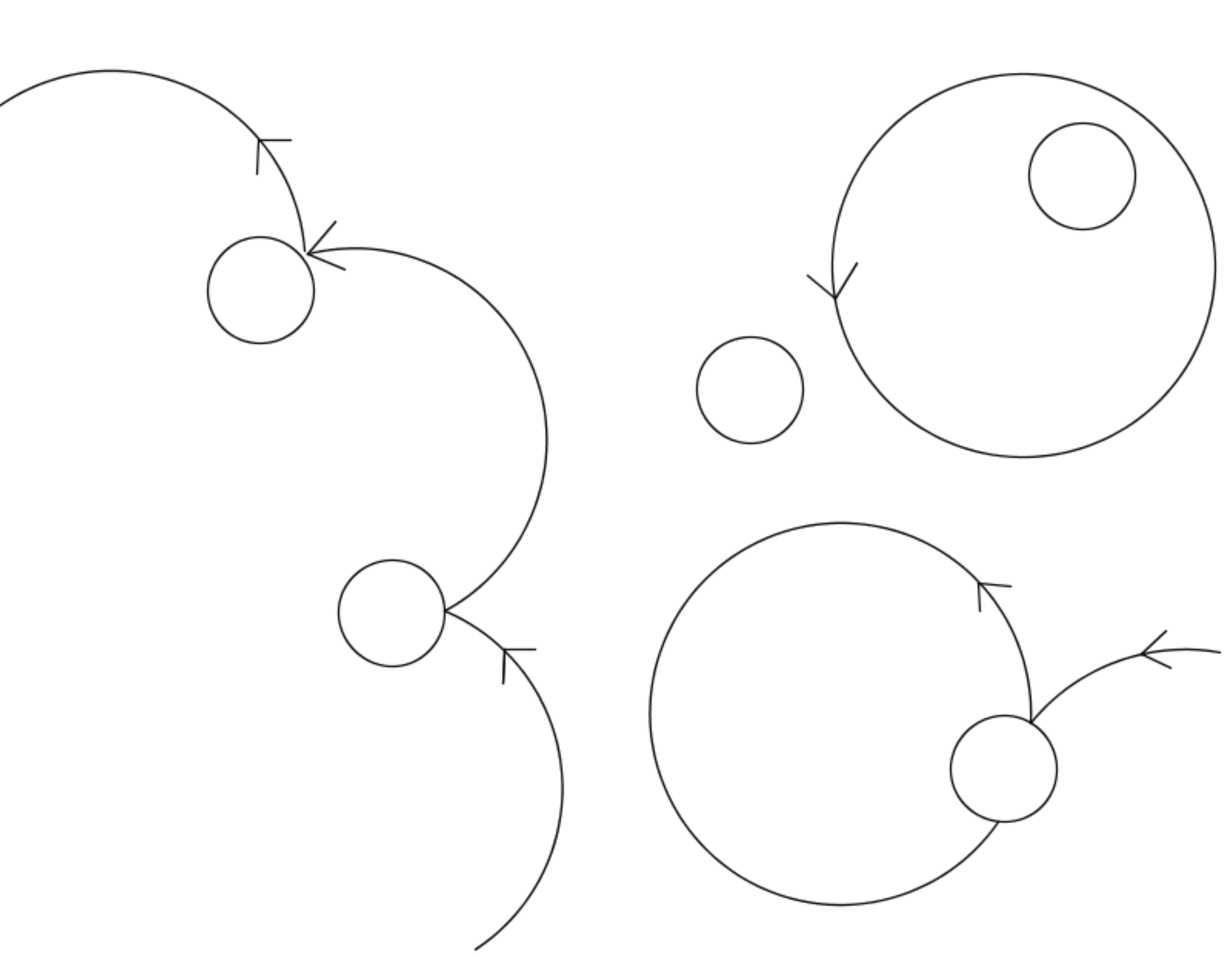} 
   \caption{Typical paths of the charged test particle when the obstacles are hard disks: due to the magnetic field, it performs arcs of circle between two consecutive collisions.}
   \label{fig:hardSphere}
\end{figure}

Furthermore it has been observed that the presence of a given external field, in the two dimensional Lorentz model, strongly affects the derivation of the linear Boltzmann equation in the Boltzmann-Grad limit. Bobylev et al, in \cite{BMHH} and later in \cite{BMHH1, BHPH} (see also \cite{KS} for further readings), showed that the set of pathological configurations is no longer negligible when the test particle moves in a plane with a Poisson distribution of hard disks and a uniform and constant magnetic field perpendicular to the plane. See Figure \ref{fig:hardSphere} for a pictorial representation of the light particle's motion. 

The following simple computation turns out to give a good heuristic argument explaining these results: consider the probability $\PP_{R_L}$ of performing an entire Larmor circle without hitting any obstacle, $R_L$ being the Larmor radius. From equation (\ref{poisson}) one easily gets  
$$\PP_{R_L}\simeq e^{-\mu_{\ep} Area(\mathcal{A}_{\ep})} \simeq e^{- 2\pi R_L\mu},$$ 
where $ \mathcal{A}_{\ep}(R_L)$ is the annulus of radius $ R_L$ and width $ \ep$. Hence, $\PP_{R_L} $ is not vanishing in the limit $\ep\to 0$ and the Markovianity of the limit system can not be attained. In fact, in \cite{BMHH1,BHPH}, a kinetic equation with memory is derived, i.e. a generalized Boltzmann equation, taking into account those effects: 
\begin{equation}\label{N:GBE}
\begin{split}
\frac{D}{Dt}f^{G}(x,v,t)=&
\mu_{\ep}\ep\sum_{k=0}^{[t/T_L]}e^{-\nu k T_L}\int_{S_1}dn\,(v\cdot n)\\&\quad\, [\chi(v\cdot n)b_n+\chi(-v\cdot n)]f^{G}(x,S_0^{-k}v,t-kT_L),
\end{split}
\end{equation}
where $f(x,v,t)$ is the probability density of finding the moving particle at time $t$ at position $x$ with velocity $v$ and
\begin{equation}
f^{G}(x,v,t)=\left\{\begin{array}{ll}
f(x,v,t)\quad\quad\quad\quad\quad\quad\quad \text{if}\;0<t<T_L&\\
(1-e^{-\nu T_L})f(x,v,t)\quad\quad \text{if}\;t>T_L.&
\end{array}\right.
\end{equation}
Here $\nu=2|v|\mu_{\ep}\ep$ is the collision frequency and $T_L= 2 \pi /\Omega$ is the cyclotron period where $ \Omega= qB /m$ is the frequency, being $q$ the charge and $m$ the mass. Furthermore, note that
$$\frac{D}{Dt}=\left(\partial_{t}+v\cdot\nabla_{x}+(v\times B)\cdot\nabla_{v}\right)$$
is the generator of the free cyclotron motion with frequency $\Omega$ and $[t/T_L]$ the number of cyclotron periods $T_L$ completed before time $t$. The angular integration over the unit vector $n$ in \eqref{N:GBE} is over the entire unit sphere $S_{1}$ centered at the origin. In the gain term the operator $b_n$ is defined by
$$b_n\phi(v)=\phi(v-2(v\cdot n) n) $$
where $\phi(v)$ is an arbitrary function of $v$. %
The precollisional velocity $v'=v-2(v\cdot n)n$ becomes $v$ after the elastic collision with the hard disk. Note that $v'\cdot n<0$. In the loss term, the precollisional velocity $v$ is also from the hemisphere $v\cdot n<0$. 
Finally, the shift operator $S_0^{-k}$, when acting on $v$, 
rotates the velocity through the angle $-k\theta$, where $\theta$ is the scattering angle (from $v'$ to $v$).

For further readings in this direction we refer to \cite{DR1,DR2}, where the authors consider a stochastic Lorentz model with a smooth external force field $ F(x,t)$ and with absorbing obstacles, i.e. the interaction between the obstacles and the test particle is such that the test particle disappears whenever it enters an obstacle. It is proved that the kinetic equation associated to this model in the Boltzmann-Grad limit is non-Markovian and that the Markovianity can be recovered by introducing an additional stochasticity in the velocity distribution of the obstacles.

\vspace{2mm} 
In this paper we consider the case of a random distribution of scatterers in $\R^2$ where each obstacle generates a smooth positive and short-range potential $\phi$, with $\alpha>0$ and sufficiently small. We show that, in this case, the solution of the microscopic dynamics converges, in the intermediate limit (when $\alpha\in (0,1/8)$), to the solution of the linear Landau equation with an additional transport term due to the magnetic field. From the heuristic point of view, this result is suggested by the observation that in this case the probability $\PP_{R_L}$ of performing an entire Larmor circle without hitting any obstacle is given by
\begin{equation*}
\mathbb{P}_{R_L} \simeq e^{-\mu_{\varepsilon}  2 \pi \varepsilon R_L} \simeq e^{-2 \pi R_L \mu \varepsilon^{-2 \alpha}}
\end{equation*} 
which vanishes as $ \varepsilon \to 0$. This computation shows that one family of the pathological events preventing the Markovianity is negligible in this setting. We stress that this rough argument is not sufficient to conclude that we can recover the Markovianity in the limit. Indeed, to prove this, we need to show that all the other bad configurations of obstacles defining the set $E$ are negligible in the limit, as we will see in Section \ref{N:pat}.

Furthermore, we observe that even if we consider a long range inverse power law interaction potential, truncated at distance $\ep^{\gamma-1}$ with $\gamma\in (0,1)$ suitably large, in the low density regime $\alpha=0$, we can prove that the memory is lost in the limit. More precisely, we prove that the microscopic solution converges to the solution of the uncutoffed linear Boltzmann equation with a magnetic transport term. With the same purpose of the rough argument presented above, we observe that 
the probability $\PP_{R_L}$ of performing a complete Larmor circle without hitting any obstacle is approximatively given by 
\begin{equation*}
\PP_{R_L}\simeq 
e^{-2\pi R_L \mu\,\ep^{\gamma-1} }\,
\end{equation*}
which vanishes as $\ep\to 0$ when $\gamma<1$. Also in this case this represents only one example of bad configuration of scatterers. It is essential to prove that the contribution of the whole set of pathological events is negligible in the limit, as we will show in Section \ref{bounds_Boltzmann}. Moreover, from the technical point of view, we observe that the parameter $\gamma$ has to be chosen close to $1$ as dictated by the explicit control of the memory effects.

Thus, as we pointed out with the heuristic motivations above, the non-Markovian behaviour of the limit process, 
discussed in \cite{BMHH}, 
disappears as soon as we slightly modify the microscopic model given by the two dimensional Lorentz Gas. 

The purpose of this paper is to provide a rigorous validation of the linear Landau equation and the linear Boltzmann equation respectively with magnetic field by using the constructive strategy due to Gallavotti. We remark that, as in \cite{DP,DR,BNP,BNPP}, we need explicit estimates of the error in the kinetic limit and this is the crucial part. Moreover, as a future target, it could be interesting to understand if a rigorous derivation of the generalized Boltzmann equation proposed in \cite{BMHH} 
can be achieved by using the same constructive techniques. 

The plan of the paper is the following: in the next Section we establish the model and formulate the results; in Section \ref{proofs} we present the strategy of the proofs, whereas Sections \ref{N:sec:ST} and \ref{bounds_Boltzmann} are dedicated to the 
nontrivial analysis and explicit estimates of the sets of bad configurations producing memory effects, which is the technical core of this paper.
\section{The Model and main results}\label{N:sec2}
\subsection{The Lorentz Model with short range interactions}
We consider the system \eqref{eq:scaled_no_B} in the plane ($d=2$) under the action of a uniform, constant, magnetic field orthogonal to the plane. The equations of motion are 
\begin{equation}\label{N:magnetic}
\left\{\begin{array}{ll}
\dot{x}=v&\\
\dot{v}=Bv^{\perp}-\ep^{\alpha-1}\sum_{i}\nabla\phi(\frac{|x-c_i|}{\ep})&,
\end{array}\right.
\end{equation}
where $B$ is the magnitude of the magnetic field and $v^{\perp}=(v_2, -v_1)$. We assume that the potential $\phi:\R^{+}\to\R^{+}$ is smooth and of range $1$ i.e. $\phi(r)=0$ if $r>1$. Therefore the particle is influenced by the scatterer $c_i$ if $|x-c_i|<\ep$.

Starting 
from the initial position $x$ with initial velocity $v$, the particle moves under the action of the Lorentz force $Bv^{\perp}$. Suppose that the particle has unitary mass and unitary charge, namely $m,q=1$, hence 
between two consecutive scatterers, the particle moves with constant angular velocity $\Omega=qB /{m}=B$ and performs an arc of circle of radius $R_L=|v|/B$. $R_L$ is the Larmor radius, i.e. the radius of the cyclotron orbit whose center is situated at the point
$$x_c=x+\frac{\mathcal{R}(\frac{\pi}{2})\cdot v}{\Omega},$$
where the tensor $\mathcal{R}(\varphi)$ denotes the rotation of angle $\varphi$. Without loss of generality we assume from now on that $ |v|=1$, therefore $ R_L = 1/B$. Moreover, we will denote by $ S_1$ the kinetic energy sphere with unitary radius.

The precise assumptions on the potential are the following:
\begin{itemize}
\item[$A1)$] $\phi\in C^2([0,\infty))$;
\item[$A2)$] $\phi\geq 0$, $\phi'\leq 0$ in $(0,1)$;
\item[$A3)$] $\supp\phi\subset [0,1]$.
\end{itemize}
On $f_0$ we assume that
\begin{itemize}
\item[$A4)$] $f_0\in C_0(\R^2\times\R^2)$ is a continuous, compactly supported initial probability density. Suppose also that $|D_{x}^kf_0|\leq C$, where $D_x$ is any partial derivative with respect to $x$ and $k=1,2$. 
\end{itemize}
Moreover, we assume that
\begin{itemize}
\item[$A5)$] The scatterers are distributed according to a Poisson distribution \eqref{poisson} of intensity $\mu_{\ep}=\mu\ep^{-\delta}$ with $\delta=1+2\alpha$, $\alpha\in(0,\frac{1}{8})$.
\end{itemize}
Next we define the Hamiltonian flow $T^t_{\mbf{c}_{N}}(x,v)$ associated to the initial datum $(x,v)$, solution of \eqref{N:magnetic} for a given configuration $\mbf{c}_{N}$ of scatterers, and we set
\begin{equation}\label{eq:microLandau}
f_{\ep}(x,v,t)=\EE_{\ep}[f_0(T^{-t}_{\mbf{c}_{N}}(x,v))]
\end{equation}
where $\EE_{\ep}$ denotes the expectation with respect to the Poisson distribution.\\
The first result of the present paper is summarized in the following theorem.
\begin{theorem}\label{N:th:main}
Let $f_{\ep}$ be defined in \eqref{eq:microLandau}.
Under assumption $A1)-A5)$, for all $t\in[0,T],\, T>0$, 
\begin{equation*}
\lim_{\ep\to 0}f_{\ep}(\cdot;t)=g(\cdot;t)
\end{equation*}
where $g$ is the unique 
solution to the Landau equation with magnetic field 
\begin{equation}\label{N:Landau}
\left\{\begin{array}{ll}\vspace{2mm}
(\partial_{t}+v\cdot\nabla_{x}+B\,v^{\perp}\cdot\nabla_{v})g(x,v,t)=\xi\,\Delta_{S_{1}} g(x,v,t)&\\
g(x,v,0)=f_0(x,v)&,
\end{array}\right.
\end{equation}
where $\Delta_{S_{1}}$ is the Laplace-Beltrami operator on the circle $S_{1}$  
and $\xi>0$. The convergence is in $L^2(\R^2\times S_{1}).$
\end{theorem}
The constant $\xi$ is the diffusion coefficient and its explicit expression will be given below in (\ref{N:coefficiente diffusione}) and in remark \ref{rm:diff_coeff}.

\subsection{The Lorentz Model with long range interactions}
We consider now the case in which each obstacle 
generates a potential of the form $$\check{\psi}_{\ep}(|x-c|)=\psi_{\ep}\big(\frac{|x-c|}{\ep}\big)$$ where the  unrescaled potential $\psi_{\ep}$ is an inverse power law potential truncated at large distances. More precisely we assume the following:

\begin{itemize}
\item[$B1)$]  $\displaystyle \qquad\quad \psi_{\ep}(x)= \left\{\begin{array}{ll}
\frac{1}{|x|^s}\qquad\quad |x|<\ep^{\gamma-1}&\\
\ep^{-s(\gamma-1)} \quad |x|\geq\ep^{\gamma-1}&
\end{array}\right.$
with $\gamma\in (0,1)$ and $s>2$. \\\vspace{0.5mm}
\end{itemize}
We point out that it could be challenging to consider directly the untruncated long range potential $ \psi(|x|)=|x|^{-s}$. In fact, this problem presents deep additional difficulties as noted in Remark 2.3 in \cite{DP} and new ideas and techniques are necessary.

Moreover, we assume that 
\begin{itemize}
\item[$B2)$] The scatterers are distributed according to a Poisson law \eqref{poisson} of intensity $\mu_{\ep}=\ep^{-1}\mu$, $\mu>0$. 
\end{itemize}
The equation of motion in macroscopic variables reads 
\begin{equation}\label{N:scaled2}
\left\{\begin{array}{ll}
\dot{x}=v&\\
\dot{v}=Bv^{\perp}-\ep^{-1}\sum_{i}\nabla\psi_{\ep}(\frac{|x-c_i|}{\ep})&
\end{array}\right.
\end{equation}
 with $\psi_{\ep}$ given in Assumption $B1)$. 

Let $T^t_{\mbf{c}_N}(x,v)$ be the Hamiltonian flow solution to equation \eqref{N:scaled2} with initial datum $(x,v)$ in a given sample of obstacles. Let be $f_0=f_0(x,v)$ be the initial probability distribution. On $f_0$ we assume
\begin{itemize}
\item[$B3)$] $f_0\in L^1\cap W^{1,\infty}(\R^2\times\R^2)$, $\quad f_0\geq 0$, $\quad \displaystyle \int f_0\,dx\,dv=1$. 
\end{itemize}
 We consider the quantity
\begin{equation}\label{N:valatt2}
f_{\ep}(x,v,t)=\EE_{\ep}[f_0(T^{-t}_{\mbf{c}_N}(x,v))]
\end{equation}
where $\EE_{\ep}$ is the expectation with respect to the measure $\PP_{\ep}$ given by \eqref{poisson}.
The second result of this paper is summarized in the following theorem.

\begin{theorem}\label{N:th:main2} 
Let $f_{\ep}$ be defined in \eqref{N:valatt2}.
Under assumption $B1)-B3)$ with $ \gamma \in (6/7,1)$, for all $t\in [0,T],\, T>0$,
\begin{equation*}
\lim_{\ep\to 0}f_{\ep}(\cdot;t)=f(\cdot;t)
\end{equation*}
where $f$ is the unique 
solution to the linear Boltzmann equation with magnetic field 
\begin{equation}\label{N:ucBoltz}
\left\{\begin{array}{ll}\vspace{2mm}
(\partial_{t}+v\cdot\nabla_{x}+B\,v^{\perp}\cdot\nabla_{v})f(t,x,v)= \text{L} f(t,x,v)&\\
f(x,v,0)=f_0(x,v),&
\end{array}\right.
\end{equation}
with
$$\displaystyle \text{L} f(v)=\mu \int_{-\pi}^{\pi}\Gamma(\Theta)\left\{f(\mathcal{R}\big(\Theta\big)v)-f(v)\right\}d\Theta\,,$$
$ \Gamma(\Theta)$ is the differential cross section associated to the long range potential $ \psi(|x|)=|x|^{-s}$ and the operator $ \mathcal{R}(\Theta)$ rotates the velocity $ v$ by the angle $ \Theta$. The convergence is in $ \mathcal{D}'(\R^2 \times S_1)$.
\end{theorem}

\section{Proofs}\label{proofs}
\subsection{Proof of Theorem \ref{N:th:main} }
Following \cite{DP,DR,BNP} we split the original problem into two parts. The first one concerns the asymptotic equivalence between $f_{\ep}$ defined in \eqref{expect_density} and $h_{\ep}$, solution of the following Boltzmann equation
\begin{equation}
\label{N:boltzmann}
\left\{\begin{array}{ll}\vspace{2mm}
(\partial_{t}+v\cdot\nabla_{x}+B\,v^{\perp}\cdot\nabla_{v})h_{\ep}(x,v,t)=\text{L}_{\ep}h_{\ep}(x,v,t)&\\
h_{\ep}(x,v,0)=f_0(x,v)&
\end{array}\right.
\end{equation}
where 
\begin{equation}
\text{L}_{\ep}h_{\ep}(v)=\mu_{\ep}\int_{-\ep}^{\ep}\,d\rho[h_{\ep}(v')-h_{\ep}(v)]\,.
\end{equation}
Here $v'=v-2(\omega\cdot v)\omega$ 
is the outgoing velocity after a scattering with incoming velocity $v$ and impact parameter $\rho\in[-\ep,\ep]$ generated by the potential $\ep^{\alpha}\phi(\frac{r}{\ep})$. Moreover, $\omega=\omega(\rho)$ is the versor bisecting the angle between the incoming and outgoing velocity and $\theta_\ep$ is the scattering angle. 
The precise result is stated in the following proposition.
\begin{proposition}\label{N:prop1LANDAU}
Under assumption $A1)-A5)$, 
for any $T>0$,
\begin{equation}
\lim _{\ep\to 0}\|f_{\ep}- h_{\ep} \|_{L^{\infty}([0,T]; L^1(\R^2\times S_1))}=0
\end{equation}
where $h_{\ep}$ solves \eqref{N:boltzmann}.
\end{proposition}
The proof of the above Proposition is postponed to Section \ref{N:sec:ST}.

The second step concerns the grazing collision limit. Note that the presence of the magnetic field does not affect the last step. More precisely we have the following 
\begin{lemma}\label{N:theta}
The deflection angle $\theta_\ep(\rho)$ of a particle colliding with impact parameter $\rho$ with a scatterer generating a radial potential $\ep^{\alpha}\phi$ under the action of the Lorentz force $Bv^{\perp}$ satisfies
\begin{equation}\label{N:est:theta}
|\theta_\ep(\rho)|\leq C\ep^{\alpha}.
\end{equation}
\end{lemma}

\begin{proof}
As established in \cite{DR} (Section 3), the estimate \eqref{N:est:theta} holds when the test particle scatters with no external field. 
Hence, we just need to compare the dynamics of the test particle in presence of the constant magnetic field with the free dynamics. Let $(\underline{x}(t),\underline{v}(t))$ be the solution of the following\begin{equation}\label{N:mtFREE}
\left\{\begin{array}{ll}
\dot{\underline{x}}=v&\\
\dot{\underline{v}}=-\ep^{\alpha-1}\nabla\phi(\frac{|\underline{x}-c|}{\ep})&.
\end{array}\right.
\end{equation}
Let $\tau$ be the collision time for the dynamics described by \eqref{N:magnetic}. 
The key observation is that the presence of the magnetic field does not modify the estimate for the collision time related to the dynamics in \eqref{N:mtFREE}.  Indeed also in our case $\tau\leq C\ep,\,C>0$, as in \cite{DR}, see Appendix \ref{app:CTime} for the detailed computations. Therefore we get 
\begin{equation*}
\begin{split}
\vert v(\tau)-\underline{v}(\tau)\vert &=
\left\vert \varepsilon^{\alpha-1}\int_{0}^{\tau}{ds\,\left(F\big(x(s)/\varepsilon\big)-F\big(\underline{x}(s)/\varepsilon\big)\right)}+\int_{0}^{\tau}{ds\,v^{\perp}B}\right\vert \\&
\leq\varepsilon^{\alpha-1}\int_{0}^{\tau}{ds\,\vert F\big(x(s)/\varepsilon\big)-F\big(\underline{x}(s)/\varepsilon\big)\vert}+C_1\varepsilon \\&
\leq\varepsilon^{\alpha-2}C_2\int_{0}^{\tau}{ds\, |x(s)-\underline{x}(s)|}+C_1\varepsilon\\&
\leq \varepsilon^{\alpha-2}C_2\int_{0}^{\tau}{ds{\int_{0}^{s}{dt\, |v(t)-\underline{v}(t)|}}}+C_1\varepsilon \\&
\leq\varepsilon^{\alpha-2}C_2\int_{0}^{\tau}{ds{\int_{0}^{\tau}{dt\, |v(t)-\underline{v}(t)|}}}+C_1\varepsilon,
\end{split}
\end{equation*}
where $ F(x):= - (\nabla \phi)(x)$. By using Gr\"{o}nwall's inequality we obtain
\begin{equation}\label{eq:asymvel}
|v(\tau)-\underline{v}(\tau)|\leq C_1\,\varepsilon\, e^{C_3\varepsilon^{\alpha-1}\tau}\leq C_1\,\varepsilon\, e^{C_3\varepsilon^{\alpha}}
\end{equation}
for $\alpha>0$ and $\ep$ sufficiently small. Hence, the velocities $v$ and $\underline{v}$ are asymptotically equivalent up to an error term of order $\ep$.
We now define $v'$ and $\underline{v}' $ to be the outgoing velocities with and without magnetic field respectively, $v$ the incoming velocity. By using \eqref{eq:asymvel} we have
\begin{align}\label{ineq:vel}
\sqrt{2(1-\cos \theta_{\ep})} = |v' - v| \leq |v' - \underline{v}' |+|\underline{v}' - v | \leq C \ep + \sqrt{2(1-\cos \tilde \theta_{\ep})}
\end{align}
where $\tilde \theta_{\ep}$ is the scattering angle without magnetic field. From \cite{DR} we know that $\tilde \theta_{\ep} \leq C' \ep^{\alpha} $, so from (\ref{ineq:vel}) we obtain 
$$ |\sin \frac{\theta_{\ep}}{2}| \leq C \ep + |\sin \frac{\tilde\theta_{\ep}}{2}| \leq C \ep + C' \ep^{\alpha} \leq C'' \ep^{\alpha}. $$ 
Since $ \theta_{\ep}$ is continuous as a function of the impact parameter $\rho$, it results $ \theta_{\ep} \leq C \ep^{\alpha}$. 
For further details see Proposition \ref{prop:thetaCR} in Appendix \ref{app:CS}.
\end{proof}
The following proposition shows the asymptotic equivalence between the solution of Landau equation and the solution of the previous Boltzmann equation $ h_{\varepsilon}$. 

\begin{proposition}\label{N:prop2LANDAU}
Under the assumptions $A1)-A5)$, $h_{\ep}\to g$ in $L^{\infty}([0,T];L^2(\R^2\times S_{1}))$ where $g$ is the unique 
solution to the Landau equation with magnetic field
\begin{equation}\label{N:Landau2}
\left\{\begin{array}{ll}
(\partial_{t}+v\cdot\nabla_{x}+B\,v^{\perp}\cdot\nabla_{v})g(x,v,t)=\xi\,\Delta_{S_{1}} g(x,v,t)&\\
g(x,v,0)=f_0(x,v)&,
\end{array}\right.
\end{equation}
where 
\begin{equation}\label{N:coefficiente diffusione}
\xi=\lim_{\ep\to 0}\frac{\mu\ep^{-2\alpha}}{2}\int_{-1}^{1}\theta_\ep^2(\rho)\,d\rho
\end{equation}
is the diffusion coefficient.
\end{proposition}

\begin{remark}\label{rm:diff_coeff}
As shown in Appendix \ref{app:CS}, $ \theta_{\ep} = \tilde\theta_{\ep} + O(\ep)$, where $ \tilde\theta_{\ep}$ is the scattering angle without any magnetic field, i.e. the scattering angle studied in \cite{DR}. This implies that the explicit expression for the diffusion coefficient obtained in \cite{DR} still holds in our case:
\begin{equation}\label{N:coefficiente diffusione2}
\xi=\frac{\mu}{2}\int_{-1}^{1}\left(\int_{\rho}^{1}{\frac{\rho}{u}\phi'(\frac{\rho}{u})\frac{\,du}{\sqrt{1-u^2}}}\right)^2\,d\rho \,,
\end{equation}
where the integrand is an even function of the impact parameter $\rho$.
\end{remark}

\begin{remark}
The linear Landau equation \eqref{N:Landau2} propagates the regularity of the derivatives with respect to the $x$ variable thanks to the transport operator. Moreover, the presence of the collision operator $\LL:=\Delta_{S_{1}}$ lets the solution gain regularity  with respect to the transverse component of the velocity. Indeed, under the assumption $A4)$ on $f_0$, the solution $g:\,\R^2\times S_{1}\to \R^{+}$ satisfies the bounds
\begin{equation}\label{bounds}
|D_{x}^{k}g |\leq C,\quad |D_{v}^{h}g(x,v)|\leq C\quad\forall k\leq 2,\,h\geq 0,
\end{equation}
$\forall t\in (0,T]$,  where $C=C(f_0,T)$ and $D_{v}$ is the derivative with respect to the transverse component of the velocity. In particular, the solutions of \eqref{N:Landau2} we are considering are classical.
\end{remark}

\begin{proof}
By using the invariance of the scattering angle with respect to the space scale, we rewrite the collision operator in the right hand side of \eqref{N:boltzmann} as
\begin{equation}\label{N:collisionboltzmann}
\text{L}_{\ep}h_{\ep}(v)=\mu_{\ep}\ep\int_{-1}^{1}\,d\rho[h_{\ep}(v')-h_{\ep}(v)].   
\end{equation}
We look at the evolution of $h_{\ep}-g$, being $g$ the solution of \eqref{N:Landau2},
namely
\begin{equation}\label{eq:evdif}
\big(\partial_t+v\cdot\nabla_x+B\,v^{\perp}\cdot\nabla_{v} \big)\big(h_{\ep}-g)
=\Big(\text{L}_\ep h_\ep-\LL g \Big),
\end{equation}
where $\LL:=\xi\,\Delta_{S_{1}}$.

Note that $ g \in 
L^2(\R^2 \times S_{1})$ because $ f_0 \in L^2(\R^2 \times \R^2)$ and $ h_{\ep} \in L^2(\R^2 \times S_{1})$. Indeed, from Proposition \ref{N:prop1LANDAU}, we know that $ h_{\ep} \in L^1(\R^2 \times S_{1})$ but the hypothesis on the initial state implies that $ h_{\ep} \in L^2(\R^2 \times S_{1}) $.

We now consider the scalar product of equation \eqref{eq:evdif} with $\big(h_{\ep}-g \big)$ in $L^2(\R^2\times S_{1})$ and we obtain 
\begin{equation*}\begin{split}
\frac 1 2 \,\partial_t \|h_{\ep}-g\|_{2}^2= &
- \big(h_{\ep}-g,\,
- \text{L}_\ep\big[h_{\ep}-g \big]  \big)\\
&+
\big(h_{\ep}-g ,\, \big[\text{L}_\ep-\LL \big] g \big).
\end{split}
\end{equation*}
By exploiting the positivity of $-\text{L}_\ep$ and the Cauchy-Schwartz inequality we get
\begin{equation*}
\partial_t \|h_{\ep}-g\|_{2}\leq  \big\|\big(\text{L}_\ep-\LL \big)\\
g\big\|_{2}.
\end{equation*}
We now set
\begin{equation*}\begin{split}
g(v')-g(v)
&=\,  (v'-v)\cdot\nabla_{|_{S_{1}}}g(v)\\
&\quad\,+\frac 1 2  (v'-v)\otimes (v'-v)\nabla_{|_{S_{1}}}\nabla_{|_{S_{1}}}g(v)\\
&\quad\,+\frac 1 6 (v'-v)\otimes (v'-v)\otimes (v'-v)
\nabla_{|_{S_{1}}}\nabla_{|_{S_{1}}}\nabla_{|_{S_{1}}}g(v)
+R_{\ep},\\
\end{split}\end{equation*}
with $R_{\ep}=\mathcal O (|v-v'|^4)$.
Integrating with respect to $\rho$ and using symmetry arguments we obtain 
\begin{equation*}\begin{split}
\text{L}_\ep g=\mu \ep^{-2\alpha}\left\{
\frac 1 2 \Delta_{S_{1}}g\int_{-1}^1 d\rho\,|v'-v|^2
+\int_{-1}^1 d\rho\,R_{\ep}\right\}.
\end{split}
\end{equation*}
Observe that $|v'-v|^2=4\sin^2\frac{\theta_\ep(\rho)}{2}$, then  
by direct computation
\begin{equation*}
\lim_{\ep\to 0}\,
\frac{\mu\ep^{-2\alpha}}{2}
\int_{-1}^1 d\rho\,|v'-v|^2=\lim_{\ep\to 0}\,
\frac{\mu\ep^{-2\alpha}}{2}\int_{-1}^{1}\theta_\ep^2(\rho)\,d\rho=:\xi.
\end{equation*}
Therefore, thanks to Lemma \ref{N:theta}, we have 
$$\big\|\big(\text{L}_\ep-\LL \big)
g\big\|_{L^2}\leq \ep^{2\alpha}\,\|\Delta^2_{|_{S_{1}}}g\|_{L^2}\leq \ep^{2\alpha}\, C,$$
which vanishes for $\ep\to 0$.

\end{proof}


\begin{remark}
We avoided introducing the cross-section $\displaystyle \Gamma(\theta_{\ep}):=\frac{\,d\rho}{\,d\theta_{\ep}}$ of the problem because the map $\rho\to\theta_{\ep}(\rho)$ is not monotonic in general. \\
Indeed if $\phi$ is bounded and $\ep$ sufficiently small, $\frac{1}{2}v^2>\ep^{\alpha}\phi(0)$ so that $\theta=0$ for $\rho=0$ and $\rho=\pm 1$. As a consequence, $\Gamma(\theta_{\ep})$ is neither single valued nor bounded.
\end{remark}
\vspace{2mm}

\subsection{Proof of Theorem \ref{N:th:main2} } 
The general structure of the proof follows the lines of [DP] where an analogous result has been proven when the magnetic field is zero.
\begin{proposition}\label{N:prop1LR}
Let $f_{\ep}$ be defined in \eqref{N:valatt2}. Then, for any $T>0$,
\begin{equation}
\lim _{\ep\to 0}\|f_{\ep}- \h \|_{L^{\infty}([0,T]; L^1(\R^2\times S_{1}))}=0
\end{equation}
where $\h$ is the unique  
solution of the truncated linear Boltzmann equation with magnetic field
\begin{equation}\label{N:cutBoltz}
\left\{\begin{array}{ll}
(\partial_{t}+v\cdot\nabla_{x}+B\,v^{\perp}\cdot\nabla_{v})\h (t,x,v)=\tilde{\text{L}}\h (t,x,v)&\\
\h(x,v,0)=f_0(x,v)&,
\end{array}\right.
\end{equation}
with
$$\displaystyle \tilde{\text{L}} f(v)=\mu \int_{-\pi}^{\pi}\Gamma^{(B)}_{\ep,\gamma}(\theta)\left\{f(\mathcal{R}\big(\theta\big)v)-f(v)\right\}d\theta.$$
and $ \Gamma^{(B)}_{\ep,\gamma}$ is the differential cross section associated to the unrescaled potential $ \psi_{\ep}$ with magnetic field.
\end{proposition}

The proof of Proposition \ref{N:prop1LR} is in Section \ref{bounds_Boltzmann}. 

This allows to reduce the problem of the transition from the solution of the truncated linear Boltzmann equation to the solution of the untruncated linear Boltzmann equation to a partial differential equation problem. Indeed, as in \cite{DP}, we can prove the following
\begin{proposition}\label{N:proplr2}
Let $\h$ solution of \eqref{N:cutBoltz}. Then, for any $T>0$,
\begin{equation}
\h \to f \quad\text{in}\quad C([0,T];\mathcal{D'})
\end{equation} 
where $f$ is the unique 
solution of \eqref{N:ucBoltz}.
\end{proposition}
\begin{proof} 
In Appendix \ref{app:CS}, Proposition \ref{RK:c-s}, the cross section $ \Gamma^{(B)}_{\varepsilon, \gamma}(\theta)$ is shown to be bounded by $ C \theta^{-1-1/s}$ and to converge to $ \Gamma(\theta)$ almost everywhere as $ \varepsilon \to 0$, where $ \Gamma(\theta)$ is the cross section associated to the truly long range potential $ \Psi(r)=r^{-s}$ without magnetic field. Therefore, the proof of Proposition 3.8 is exactly the same as the one of Proposition A.2 in \cite{DP}.
\end{proof}

\section{Proof of Proposition \ref{N:prop1LANDAU}}\label{N:sec:ST}
In this Section we prove the asymptotic equivalence of $f_\ep$, defined by \eqref{eq:microLandau}, and $h_\ep$, solution of the linear Boltzmann equation \eqref{N:boltzmann}, that we recall here for the sake of clarity
\begin{equation}
\label{N:Boltzmann epsilon}
(\partial_t+v\cdot\nabla_x+B\,v^{\perp}\cdot\nabla_v)h_{\ep}(x,v,t)=\text{L}_{\ep}h_{\ep}(x,v,t),
\end{equation}
where 
\begin{equation}
\text{L}_{\ep}h_{\ep}(v)=\mu\ep^{-2\alpha}\int_{-1}^{1}\,d\rho\{h_{\ep}(v')-h_{\ep}(v)\}.
\end{equation}
This allows to reduce the problem to the analysis of a Markov process which is an easier task. Indeed, the series expansion defining $h_{\ep}$ (obtained perturbing around the loss term) reads as
\begin{equation}
\label{eq:h_eps}
h_{\ep}(x,v,t)=e^{-2\ep^{-2\alpha}\mu t}\sum_{Q\geq 0}\mu_{\ep}^{Q}\int_{0}^{t}dt_Q\dots\int_{0}^{t_{2}}dt_1
\int_{-\ep}^{\ep}d\rho_1\dots\int_{-\ep}^{\ep}d\rho_Q\, f_0(\gamma^{-t}(x,v)).
\end{equation}
Here $\gamma^{-t}(x,v)=({\xi}_{\ep}(-t),{\eta}_{\ep}(-t))$ where ${\eta}_{\ep}$ is an autonomous jump process and
${\xi}_{\ep}$ is an additive functional of ${\eta}_{\ep}$. Equation \eqref{eq:h_eps} is an evolution equation for the probability density associated to a particle performing random jumps in the velocity space at random Markov times. 

We start the proof by looking at the microscopic solution $f_\ep$ defined by  \eqref{expect_density}. For $(x,v)\in\R^2\times\R^2$, $t>0$, we have
\begin{equation}
\label{N:formula1}
f_{\ep}(x,v,t)=e^{-\mu_{\ep}|\mathcal{B}(x,t)|}\sum_{N\geq 0}\frac{\mu_{\ep}^{N}}{N!}\int_{\mathcal{B}(x,t)^N}d\mbf{c}_{N}\, f_0(T^{-t}_{\mbf{c}_{N}}(x,v)),
\end{equation}
where $T^t_{\mbf{c}_{N}}(x,v)$ is the Hamiltonian flow with initial datum $(x,v)$ and $\mathcal{B}(x,t)$ 
is the disk centered in $x$ with radius $t$.

Given the configuration of obstacles $\mbf{c}_{N}=c_1\dots c_N$, we shall say that $c_i$ is {\it internal} if it influences the motion up to the time $t$, i.e. 
\begin{equation}
\inf_{0\leq s\leq t}|x_{\ep}(-s)-c_i|<\ep,
\end{equation}
while we shall call $c_i$ {\it external} if 
\begin{equation}
\inf_{0\leq s\leq t}|x_{\ep}(-s)-c_i|\geq\ep.
\end{equation}
Here $(x_{\ep}(-s),v_{\ep}(-s))=T_{\mbf{c}_{N}}^{-s}(x,v)$, $s\in [0,t]$.

We can perform the integration over the external obstacles and we get
\begin{equation}
\label{N:formula2}
f_{\ep}(x,v,t)=\,\sum_{Q\geq 0}\frac{\mu_{\ep}^{Q}}{Q!}\int_{\mathcal B(x,t)^Q}d\mbf{b}_{Q}\, e^{-\mu_{\ep}|\T(\mbf{b}_{Q})|}
 \chi(\{\mbf{b}_{Q}\;\text{internal}\})f_0(T^{-t}_{\mbf{b}_{Q}}(x,v)),
\end{equation}
where $\chi(E)$ is the characteristic function of the event $E$ and $\T(\mbf{b}_{Q})$ is the tube
\begin{equation}
\T(\mbf{b}_{Q})=\{y\in \mathcal B(x,t)\;\text{s.t.}\;\exists s\in(0,t)\;\text{s.t.}\;|y-x_{\ep}(-s)|<\ep\}.
\end{equation}
Note that in the previous integration we are not considering possible overlappings of obstacles. This is legitimate because we shall see that this event is negligible as $ \ep$ tends to 0.

Furthermore, let us restrict to the configurations such that the light particle's trajectory does not start from inside an obstacle and does not end inside an obstacle: in formula
\begin{equation}
\chi_1(\mbf{b}_{Q})=\chi\{\mbf{b}_{Q}\;\text{s.t.}\; b_i\notin \mathcal B(x,\ep)\;\text{and}\; b_i\notin \mathcal B(x_\ep(-t),\ep)\;\text{for}\;\text{all}\; i=1,\dots,Q\}.
\end{equation}
As for the overlappings, this choice is not really restrictive because the contribution related to $ 1-\chi_1$ is going to vanish in the limit, as we shall see.
Moreover, we will now list other events that will turn out to be negligible as $\ep $ approaches 0.
\begin{itemize}
\item[i)] \textbf{Complete cyclotronic orbit}.\\ A first cyclotron orbit is completed without suffering any collisions and a repeated collision occurs with the same scatterer without any collision in the meantime.
\end{itemize}
\begin{figure}
\centering
\includegraphics[scale=0.10]{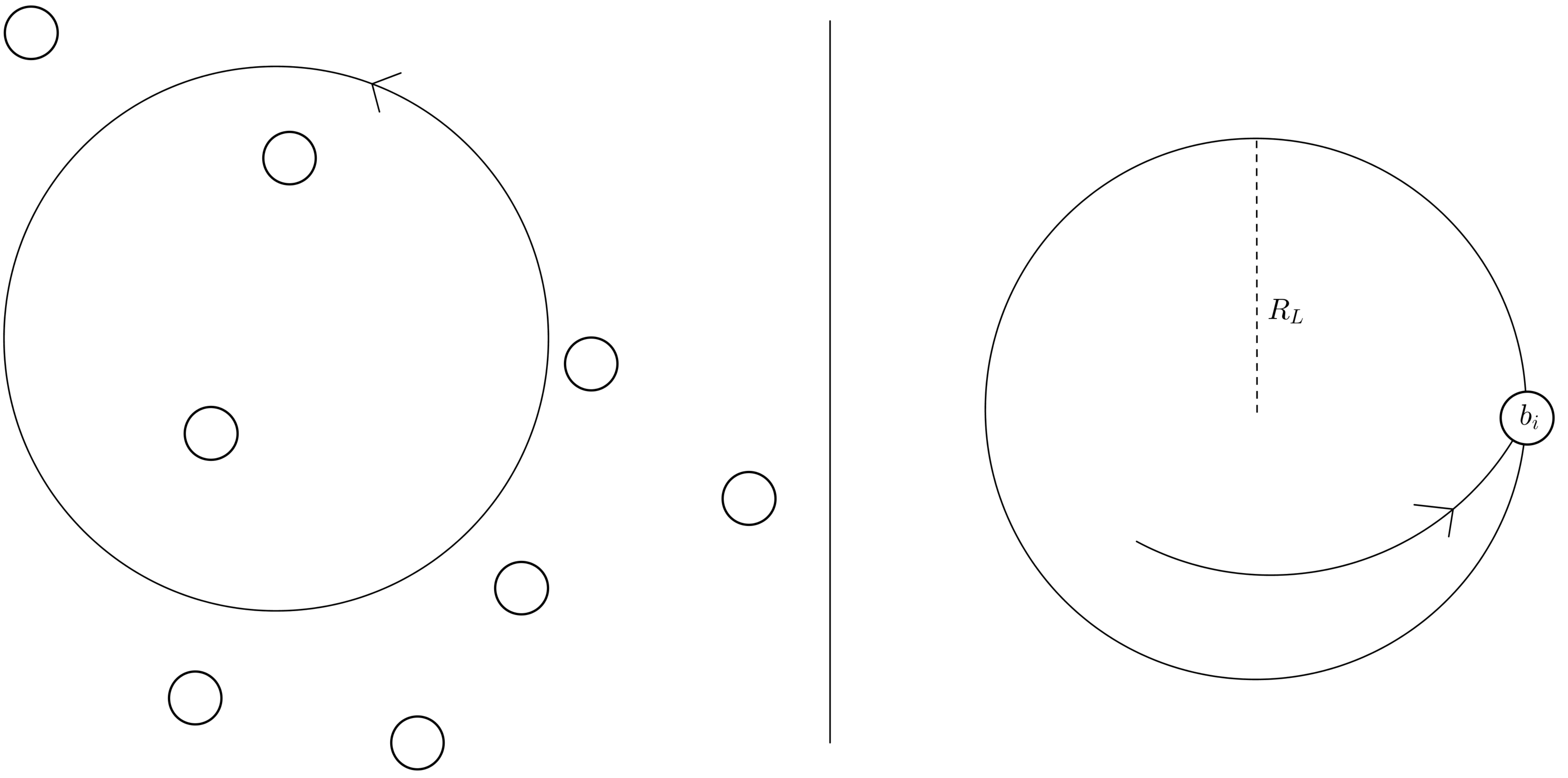}
\caption[size=0.08]{On the left a cyclotron orbit completed without suffering collisions is represented. On the right there is a repeated collision with the same scatterer. 
}\label{N:fig:1M}
\end{figure}
We set
\begin{equation}\label{N:chicirc}
\chi_{circ}(\mbf{b}_{Q})=\chi\left(\left\{\mbf{b}_{Q}\;\,\text{s.t.}\;i)\,\text{is\,realized} \right\}\right).
\end{equation}
For a pictorial representation of the event i) see Figure \ref{N:fig:1M}.
We now define
\vspace{1.7mm}
\begin{equation}
\begin{split}
\label{N:eq:f3}
\breve{f}_{\ep}(x,v,t)=\,&\sum_{Q\geq 0}\frac{\mu_{\ep}^{Q}}{Q!}\int_{\mathcal B_(x,t)^Q}d\mbf{b}_{Q}e^{-\mu_{\ep}|\T(\mbf{b}_{Q})|}\chi(\{\mbf{b}_{Q}\;\text{internal}\})\\&
\quad  (1-\chi_{circ}(\mbf{b}_{Q}))\,\chi_1(\mbf{b}_{Q})f_0(T^{-t}_{\mbf{b}_{Q}}(x,v)).
\end{split}
\end{equation}
Note that $f_{\ep}\geq\breve{f}_{\ep}$. For $t <T_L$ one expects that the approximation with the dynamics of the test particle in absence of the external field is true. The unexpected fact is that even for $t \gtrsim O(T_L) $ this still holds because \eqref{N:chicirc} tends to 0 as $ \ep \to 0$. Hence, for a given configuration $\mbf{b}_{Q}$ such that $\chi_1 [1-\chi_{circ}](\mbf{b}_{Q})=1$, we have that the measure of the tube can be estimated by
\begin{equation}\label{N:estimateT}
|\T(\mbf{b}_{Q})|\leq 2\ep t. 
\end{equation}

At this point we define
\begin{equation}
\begin{split}
\label{N:formula3}
\tilde{f}_{\ep}(x,v,t)=\,&e^{-2\mu\ep^{-2\alpha}t}\sum_{Q\geq 0}\frac{\mu_{\ep}^{Q}}{Q!}\int_{\mathcal B(x,t)^Q}d\mbf{b}_{Q}\chi(\{\mbf{b}_{Q}\;\text{internal}\})\\&
\quad (1-\chi_{circ}(\mbf{b}_{Q}))\,\chi_1(\mbf{b}_{Q})f_0(T^{-t}_{\mbf{b}_{Q}}(x,v)).
\end{split}
\end{equation}
Thanks to \eqref{N:estimateT} we get
\begin{equation}\label{eq:monotonicity1}
f_{\ep}\geq \breve{f}_{\ep}\geq\tilde{f}_{\ep}.
\end{equation}
According to a classical argument introduced in \cite{G} (see also \cite{DP,DR,BNP}), we now want to remove from $\tilde{f}_{\ep}$ all the events that prevent the light particle's trajectory to be the Markov process described by $h_{\ep}$. 

For any fixed initial condition $(x,v)$ we order the obstacles $b_1,\dots,b_N$ according to the scattering sequence. Let $\rho_i$ and $t_i$ be the impact parameter and the backwards entrance time of the light particle in the protection disk around $b_i$, namely $\mathcal B (b_i, \ep)$.
Then we perform the following change of variables
\begin{equation}
\label{N:change var}
b_1,\dots,b_N\rightarrow \rho_1,t_1,\dots,\rho_N,t_N
\end{equation}
with 
\begin{equation*}
0\leq t_N<t_{N-1}<\dots<t_1\leq t.
\end{equation*}
Conversely, fixed the impact parameters $\{\rho_i\}$ and the hitting times $\{t_i\}$ we construct the centers of the obstacles $b_i=b(\rho_i,t_i)$ and a trajectory $\bar\gamma^{-s}(x,v):=(\bar\xi_{\ep}(-s),\bar\eta_{\ep}(-s)),\ s\in[0,t]$ inductively. 

Suppose that we are able to define the obstacles $b_1, \ldots, b_{i-1}$
and   a   trajectory $\bar\gamma^{-s}(x,v):=(\bar\xi_{\ep}(-s),\bar\eta_{\ep}(-s))$ up to the time $s=t_{i-1}$. We then define the trajectory between times
$t_{i-1}$ and $t_{i}$ as that of the evolution of a particle moving under the action of the Lorentz force and of the potential $\ep^{\alpha}\phi({\ep}^{-1}{|\cdot-b_{i-1}|})$ with initial datum $(\bar\xi_{\ep}(-t_{i-1}),\bar\eta_{\ep}(-t_{i-1})).$ 
Then $b_i$ is defined to be the only point at distance $\ep$ of $\bar\xi_{\ep}(-t_i)$ and algebraic distance $\rho_i$ from the straight line which is tangent to the trajectory at the point $\bar\xi_{\ep}(-t_i)$.

However, $\bar\gamma^{-s}(x,v)=(x_{\ep}(-s),v_{\ep}(-s))$ (therefore the mapping \eqref{N:change var} is one-to-one) only outside the following pathological situations (relative to the backward trajectory). 
\begin{itemize}
\item[ii)] \textbf{Overlapping}.\\ If $b_i$ and $b_j$ are both internal then $\mathcal B(b_i,\ep)\cap \mathcal B(b_j,\ep)\neq\emptyset$.
\item[iii)] \textbf{Recollisions}.\\ There exists $b_i$ such that for $s\in(t_{j+1},t_{j})$, $j>i$, $\xi_{\ep}(-s)\in \mathcal B(b_i,\ep)$. 
\item[iv)] \textbf{Interferences}.\\ There exists $b_i$ such that $\xi_{\ep}(-s)\in \mathcal B(b_j,\ep)$ for $s\in(t_{i+1},t_{i})$, $j>i$.
\end{itemize}

We simply skip such events by setting
\begin{equation*}
\chi_{ov}=\chi(\{\mbf{b}_Q\;\text{s.t.}\;\text{ii)}\;\text{is}\;\text{realized}\}),
\end{equation*}
\begin{equation*}
\chi_{rec}=\chi(\{\mbf{b}_Q\;\text{s.t.}\;\text{iii)}\;\text{is}\;\text{realized}\}),
\end{equation*}
\begin{equation*}
\chi_{int}=\chi(\{\mbf{b}_Q\;\text{s.t.}\;\text{iv)}\;\text{is}\;\text{realized}\}),
\end{equation*}
and defining
\begin{equation}
\begin{split}\displaystyle
\label{N:formula4}
\bar{f}_{\ep}(x,v,t)=\,& e^{-2\ep^{-2\alpha}\mu t}\sum_{Q\geq 0}\mu_{\ep}^{Q}\int_{0}^{t}dt_1\dots\int_{0}^{t_{Q-1}}dt_Q\int_{-\ep}^{\ep}d\rho_1\dots\int_{-\ep}^{\ep}d\rho_Q
\\&\quad
\,\chi_1(1-\chi_{circ}) (1-\chi_{ov})(1-\chi_{rec})(1-\chi_{int})f_0(\bar\gamma^{-t}(x,v)).
\end{split}
\end{equation}
Note that 
\begin{equation}\label{eq:monotonicity2}
\bar{f}_{\ep}\leq\tilde{ f}_{\ep}\leq \breve{f}_{\ep}\leq f_{\ep}.
\end{equation} 
Note also that in \eqref{N:formula4} we have used the change of variables \eqref{N:change var} for which, outside the pathological sets i), ii), iii), iv) $\bar\gamma^{-t}(x,v)=(x_{\ep}(-t),v_{\ep}(-t))$.

Next we remove $\chi_1(1-\chi_{circ})(1-\chi_{ov})(1-\chi_{rec})(1-\chi_{int})$ by setting 
\begin{equation}
\begin{split}
\label{N:formula5}
\bar{h}_{\ep}(x,v,t)=&\,e^{-2\ep^{-2\alpha}\mu t}\sum_{Q\geq 0}\mu_{\ep}^{Q}\int_{0}^{t}dt_1\dots\int_{0}^{t_{Q-1}}dt_Q \\&
\quad\,
\int_{-\ep}^{\ep}d\rho_1\dots\int_{-\ep}^{\ep}d\rho_Q\, f_0(\bar\gamma^{-t}(x,v)).
\end{split}
\end{equation}
We observe that 
\begin{equation}
\label{N:error}
1-\chi_1(1-\chi_{ov})(1-\chi_{cir})(1-\chi_{rec})(1-\chi_{int}) \leq (1-\chi_1)+\chi_{ov}+\chi_{cir}+\chi_{rec}+\chi_{int}.
\end{equation}
Then by \eqref{N:formula4} and \eqref{N:formula5} we obtain
$$|\bar{h}_{\ep}(t)-\bar{f}_{\ep}(t)|\leq\ffi_1(\ep,t)$$
with
\begin{equation}
\begin{split}\displaystyle
\label{N:ffi_1}
\ffi_1(\ep,t)=\,& \|f_0\|_{\infty}\,e^{-2\ep^{-2\alpha}\mu t}\sum_{Q\geq 0} \mu_{\ep}^{Q}\int_{0}^{t}dt_1\dots\int_{0}^{t_{Q-1}}dt_Q\int_{-\ep}^{\ep}d\rho_1\dots\int_{-\ep}^{\ep}d\rho_Q
\\&\quad
\,[(1-\chi_1)+\chi_{ov}+\chi_{cir}+\chi_{rec}+\chi_{int}].
\end{split}
\end{equation}
We can prove that $ \varphi_1$ is negligible in the limit. The precise statement follows.  
\begin{proposition}\label{N:th:prop}
Let $\ffi_1(\ep,t)$ be defined as in \eqref{N:ffi_1}. For any $t\in[0,T]$
\begin{equation*}
\|\ffi_1(\ep,t)\|_{L^1}\to 0
\end{equation*}
as $\ep\to 0$.
\end{proposition}
\begin{proof}
See Section \ref{N:pat}.
\end{proof}

To complete the proof of Theorem \ref{N:th:main} we still need to show the asymptotic equivalence of $\bar{h}_{\ep}$ and $h_{\ep}$. Notice that $ h_{\ep}$ is given by (\ref{eq:h_eps}) where the trajectory $ \gamma^{-t}(x,v)=(\xi_{\ep}(-t),\eta_{\ep}(-t))$ is a jump process in the velocity space, i.e. the changes of velocity are instantaneous. 
We compare the trajectory $ \bar\gamma^{-t}(x,v)=(\bar\xi_{\ep}(-t),\bar\eta_{\ep}(-t))$ with $ \gamma^{-t}(x,v)$: being  $t_1,\dots,t_Q$ the sequence of impact times and $ \tau \leq C \ep$ the collision time, the spacial coordinates can differ only inside the interaction disk, while the velocities can differ only if $t \in (t_1, t_1 +\tau) $. 

In formulae we have 
\begin{align}
|\xi_{\ep}(-t)-\bar{\xi_{\ep}}(-t)| &\leq C_1 Q\,\ep \nonumber \\
|\eta_{\ep}(-t)-\bar{\eta_{\ep}}(-t)|& \leq C_2 \ep^{\alpha} \chi(t -t_1 \leq C \ep)\,.   
\end{align}
By exploiting the regularity of the initial condition $ f_0$ we get
\begin{align}
|f_0(\xi_{\ep}(-t),\eta_{\ep}(-t)) - f_0(\bar\xi_{\ep}(-t),\bar\eta_{\ep}(-t))| \leq C'[Q \ep + \ep^{\alpha} \chi(t -t_1 \leq C \ep)]
\end{align}
which implies 
\begin{align}
|h_{\ep}(x,v,t) - \bar{h}_{\ep}(x,v,t) | \leq & \, C' e^{-2\mu t\ep^{-2\alpha}}\sum_{Q\geq 0}\mu_{\ep}^{Q}\int_{0}^{t}dt_1\dots\int_{0}^{t_{Q-1}}dt_Q \nonumber \\
& \times  \int_{-\ep}^{\ep}d\rho_1\dots\int_{-\ep}^{\ep}d\rho_Q\, [Q \ep + \ep^{\alpha} \chi(t -t_1 \leq C \ep)] \nonumber \\
 \leq & \, C(t \ep^{1-2\alpha} + \ep^{1+\alpha})\,.
\end{align}
Hence we obtained $\displaystyle \lim_{\ep \to 0}\norm{h_{\ep} - \bar{h}_{\ep}}_{L^{\infty}([0,T] \times \R^2 \times S_{1})}=0 $.
We observe that the monotonicity argument behind this strategy, see equation \eqref{eq:monotonicity2}, the positivity of the solution $h_{\ep}$ of the Boltzmann equation and the conservation of mass imply that $\bar{f}_{\ep}, \bar{h}_{\ep}$ and $h_{\ep}$ have the same asymptotic behavior in $L^{\infty}\left([0,T]; L^{1}(\R^2 \times S_{1})\right)$ when $\ep\to 0$.

\subsection{Control of the pathological sets: proof of Proposition \ref{N:th:prop}} \label{N:pat}
In this section we prove Proposition \ref{N:th:prop}. This makes rigorous the claim of the heuristic argument presented in the paper's introduction.
For any measurable function $u$ of the backward Markov process $(\xi_{\ep},\eta_{\ep})$ 
we set
\begin{equation*}
\begin{split}
\EE _{x,v}[u]&=e^{-2 \mu_{\ep}\ep t} \sum_{Q\geq 0}\mu_{\ep}^Q\int_{0}^{t}dt_1\dots\int_{0}^{t_{Q-1}}dt_Q\\&
\int_{-\ep}^{\ep}d\rho_1\dots\int_{-\ep}^{\ep}d\rho_Q \,u(\xi_{\ep},\eta_{\ep}).
\end{split}
\end{equation*}
Recalling \eqref{N:error} we have
$$\ffi_1(\ep,t) \leq \|f_0\|_{\infty}\EE _{x,v}[(1-\chi_1)+\chi_{ov}+\chi_{cir}+\chi_{rec}+\chi_{int}]\,.$$
We can skip the estimates of the first two contributions, i.e. $\EE _{x,v}[(1-\chi_1)]$ and $\EE _{x,v}[\chi_{ov}]$, since the presence of the external field does not affect the classical arguments which can be found in \cite{BNP,DR,DP}.  However, the presence of the magnetic field and consequently the circular motion of the test particle strongly affects the explicit estimates of the pathological events ii), iii). Therefore, we need a detailed analysis for $\chi_{cir}$, $\chi_{rec}$ and $\chi_{int}$.

For what concerns the pathological event due to a recollision with the same scatterer (see Figure \ref{N:fig:2M}), we observe that $\chi_{cir}=1$ if there exists an entrance time $t_i$ such that $|t_{i}-t_{i+1}|\geq T_L-\tau\geq T_L-C\varepsilon$ for some $i=0,\dots Q-1$.  Moreover, $ \chi_{circ}=1$ also when a test particle performs an entire Larmor orbit without colliding with any obstacles. As explained in the introduction, the probability of this event is bounded from above by $ C \exp(- \frac{2 \pi \mu}{B}\ep^{-2 \alpha}):= c_{\alpha}(\ep)$. 
Therefore, it results
\begin{figure}
\centering
\includegraphics[scale=0.1]{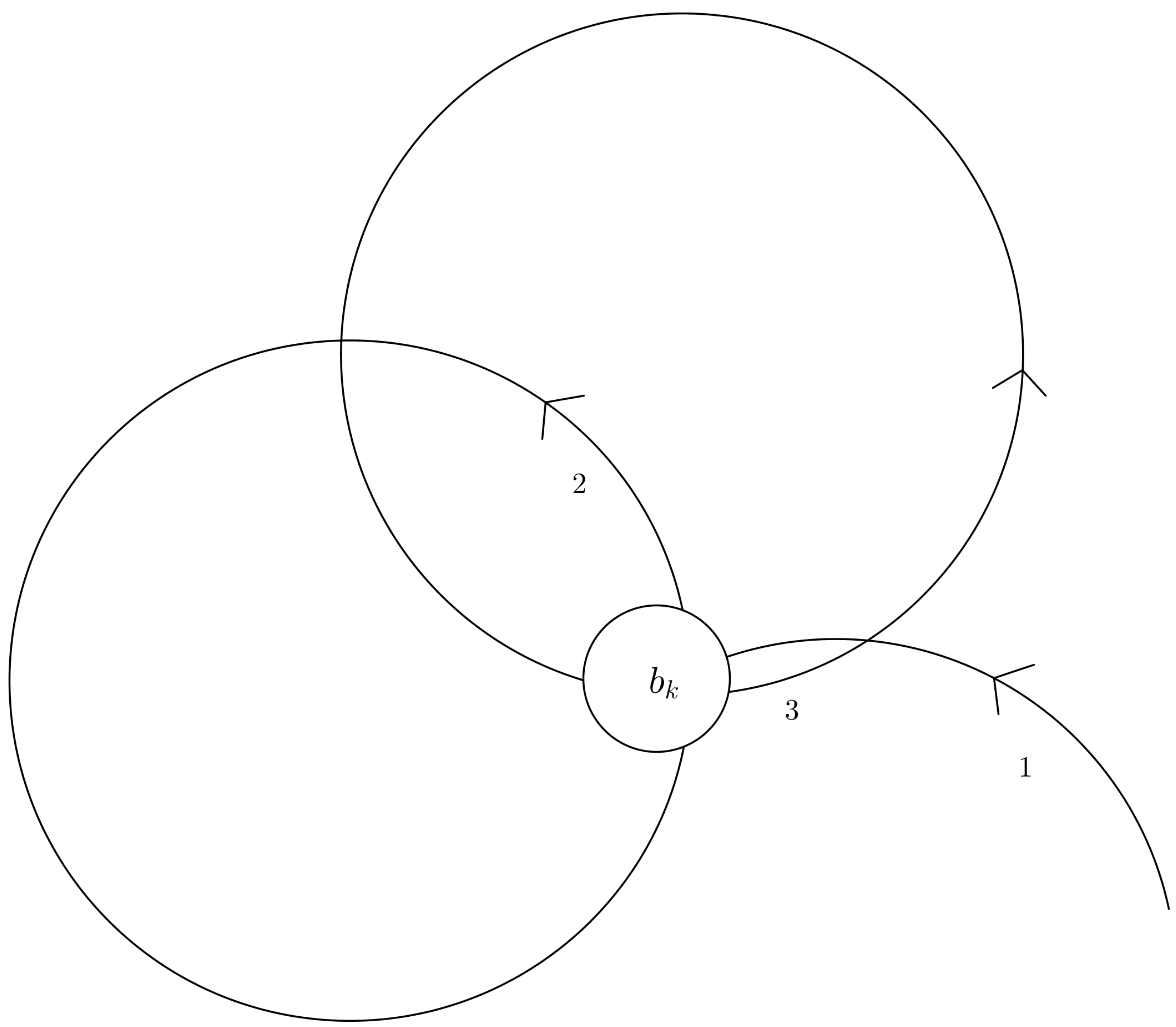}
\caption{Recollision with the same scatter.}
\label{N:fig:2M}
\end{figure}


\begin{align}
\EE_{x,v}[\chi_{circ}]&\leq c_{\alpha}(\ep) + e^{-2 \mu_{\ep}\ep t} \sum_{Q\geq 1} \mu_{\ep}^Q\int_{0}^{t}dt_1\dots\int_{0}^{t_{Q-1}}dt_Q 
\int_{-\ep}^{\ep}d\rho_1\dots\int_{-\ep}^{\ep}d\rho_Q \, \sum_{i=1}^{Q}\chi(t_{i}<t_{i-1} - T_L+C\varepsilon) \nonumber \\&
\leq c_{\alpha}(\ep) + e^{-2t \mu_{\varepsilon}\varepsilon }\sum_{Q\geq 1}(2 \mu_{\varepsilon}\varepsilon)^Q \int_0^t dt_1 \cdots \int_0^{t_{Q-1}}dt_Q \sum_{i=1}^{Q}\chi(t_{i}<t_{i-1} - T_L+C\varepsilon) 
\end{align}
where $ \mu_{\varepsilon}=\mu \varepsilon^{-1-2\alpha}$ with $ \alpha \in (0,1/2)$ and $ t_0=t$.

We set 
\begin{align*}
I_i &:= \int_0^t dt_1 \cdots \int_0^{t_{Q-1}}dt_Q \chi(t_{i}<t_{i-1} - T_L+C\varepsilon) \nonumber \\
& =\int_0^t dt_1 \cdots \int_0^{t_{i-1}}dt_i \chi(t_{i}<t_{i-1} - T_L+C\varepsilon)\frac{t_i^{Q-i}}{(Q-i)!} \nonumber \\
&= \int_0^t dt_1 \cdots \int_0^{t_{i-2}}dt_{i-1} \frac{(t_{i-1}-T_L + C\varepsilon)^{Q-(i-1)}}{(Q-(i-1))!}
\end{align*}
then 
\begin{align}
\EE_{x,v}[\chi_{circ}]\leq c(\ep) + e^{-2t \mu_{\varepsilon}\varepsilon }\sum_{Q\geq 1}(2 \mu_{\varepsilon}\ep)^Q \sum_{i=1}^{Q}I_i\,.
\end{align}
Note that 
\begin{align*}
I_{i+1} &=\int_0^t dt_1 \cdots \int_0^{t_{i-1}}dt_{i} \frac{(t_{i} - T_L+C\varepsilon)^{Q-i}}{(Q-i)!}  \nonumber \\
&= \int_0^t dt_1 \cdots \int_0^{t_{i-2}}dt_{i-1} \frac{(t_{i-1}-T_L + C\varepsilon)^{Q-(i-1)}-(-T_L+C\varepsilon)^{Q-(i-1)}}{(Q-(i-1))!} \nonumber \\
&= I_i -\frac{(-T_L + C\varepsilon)^{Q-i+1}}{(Q-i+1)!}\frac{t^{i-1}}{(i-1)!} \nonumber \\
&= I_i - f_{i-1}
\end{align*}
where $ f_i:=\frac{(-T_L + C\varepsilon)^{Q-i}}{(Q-i)!}\frac{t^{i}}{(i)!}$. We now look at\begin{align}\label{eq:Ii}
\sum_{i=1}^Q I_i 
&= QI_1-\sum_{j=0}^{Q-2}(Q-1-j)f_j \nonumber \\
&=\frac{1}{(Q-1)!}\bigg[(t-T_L + C\varepsilon)^Q - \sum_{j=0}^{Q-2}(Q-1)!(Q-j)\frac{(-T_L+C\varepsilon)^{Q-j}}{(Q-j)!}\frac{t^j}{j!} +\sum_{j=0}^{Q-2}(Q-1)!\frac{(-T_L+C\varepsilon)^{Q-j}}{(Q-j)!}\frac{t^j}{j!}  \bigg] \nonumber \\
&\leq \frac{1}{(Q-1)!} \bigg[ 2(t-T_L + C\varepsilon)^Q +(T_L-C\varepsilon)(t-T_L+C\varepsilon)^{Q-1} \bigg].\nonumber \\
\end{align}
Finally we got
\begin{align}
\EE_{x,v}[\chi_{circ}] \leq & \, c_{\alpha}(\ep) + e^{-2t \mu_{\varepsilon}\varepsilon }\sum_{Q\geq 1}(2 \mu_{\varepsilon}\varepsilon)^Q \bigg[ \frac{2(t-T_L + C\varepsilon)^Q}{(Q-1)!} + \frac{(T_L-C\varepsilon)(t-T_L+C\varepsilon)^{Q-1}}{(Q-1)!}\bigg] \nonumber \\
= & \,c_{\alpha}(\ep) + 2e^{-2t \mu_{\varepsilon}\varepsilon }(2 \mu_{\varepsilon}\varepsilon)(t-T_L + C\varepsilon)\sum_{Q\geq 1}(2 \mu_{\varepsilon}\varepsilon)^{Q-1} \frac{(t-T_L + C\varepsilon)^{Q-1}}{(Q-1)!} \nonumber \\
& + e^{-2t \mu_{\varepsilon}\varepsilon }(T_L-C\varepsilon)(2 \mu_{\varepsilon}\varepsilon)\sum_{Q\geq 1} (2 \mu_{\varepsilon}\varepsilon)^{Q-1}\frac{(t-T_L+C\varepsilon)^{Q-1}}{(Q-1)!} \nonumber \\
= & \,c_{\alpha}(\ep) + 2e^{-2(T_L-C\varepsilon) \mu \varepsilon^{-2 \alpha} }(2 \mu\varepsilon^{-2 \alpha})(t-T_L + C\varepsilon) + e^{-2(T_L-C\varepsilon) \mu \varepsilon^{-2 \alpha} }(T_L-C\varepsilon)(2 \mu \varepsilon^{-2 \alpha}) 
\end{align}
for $\alpha>0$ and $\ep$ sufficiently small. Hence, $\EE_{x,v}[\chi_{circ}]$ vanishes as $ \varepsilon \to 0$. 

We now consider a generalization of $\chi_{cir}$: let be $\chi_{arc}^{(\nu)}$ the characteristic function of the event such that the light particle does not
hit any obstacles in a time interval equal to $T_L\,\ep^{\nu},$ $0<\nu< 1$. 
More precisely $\chi_{arc}^{(\nu)}=1$ if there exists an entrance time $t_i$ such that $|t_{i}-t_{i+1}|\geq T_L\,\ep^{\nu}-\tau\geq T_L\,\ep^{\nu}-C\varepsilon$ for some $i=0,\dots Q-1.$
The same computations as for $\EE_{x,v}[\chi_{circ}]$ show that $\EE_{x,v}[\chi_{arc}^{(\nu)}]$ vanishes as $ \varepsilon \to 0$ when $\nu< 2\alpha.$ 
In other words this shows that the motion of the light particle outside the obstacles covers arcs of circle and corresponding angles of order at most $O(\ep^{\nu})$. 


Next we pass to the control of the recollision event. We observe that
\begin{equation}\label{eq:arcrec}
\chi_{rec}=\big(1-\chi_{arc}^{(\nu)}\big)\chi_{rec}+\chi_{rec} \, \chi_{arc}^{(\nu)}\leq  \big(1-\chi_{arc}^{(\nu)}\big)\chi_{rec}+\chi_{arc}^{(\nu)}
\end{equation}
and this implies
$$\EE_{x,v}[\chi_{rec}]\leq \EE_{x,v}[\big(1-\chi_{arc}^{(\nu)}\big)\chi_{rec}]+\EE_{x,v}[\chi_{arc}^{(\nu)}]\,,$$
but $\EE_{x,v}[\chi_{arc}^{(\nu)}]$ is vanishing in the limit $ \varepsilon \to 0$ as we have seen before. Therefore, we can focus on $\EE_{x,v}[\big(1-\chi_{arc}^{(\nu)}\big)\chi_{rec}]$.
Let $t_i$ the first time the light particle hits the $ i$-th scatterer $ b_i$, $v_i^{-}$ the incoming velocity, $v_i^{+}$ the outgoing velocity (with respect to the backwards trajectory) and $t_i^{+}$ the exit time. Moreover, we fix the axis in such a way that $v_i^{+}$ is parallel to the $x$ axis. 
We have
\begin{equation} 
\chi_{rec}\big(1-\chi_{arc}^{(\nu)}\big)\leq \big(1-\chi_{arc}^{(\nu)}\big)\sum_{i=1}^{Q}\sum_{j>1}\chi_{rec}^{i,j},
\end{equation}
where  $\chi_{rec}^{i,j}=1$ if and only if $b_i$ (constructed via the sequence $t_1,\rho_1,\dots,t_i,\rho_i$) is recollided in the time interval $(t_j,t_{j-1})$. Note indeed that a recollision can occur only if the rotation angle $|\sum_{h=i+1}^{j-1}(\theta_h+\varphi_h)|>\pi$ where $\varphi_h$ is the angle covered outside the obstacles in the time interval $(t_{h+1},t_{h}^{+})$, being 
$\theta_h$ the $ h$-th scattering angle. The constraint $\big(1-\chi_{arc}^{(\nu)}\big)$ implies that $|\varphi_h|\leq C'\ep^{\nu}$.\\
Hence, since $|\theta_h+\varphi_h|\leq C\,\ep^{\alpha}+C' \,\ep^{\nu}\leq C''\ep^{\nu/2},\; $ in order to have a recollision there must be an intermediate velocity $v_k$, $k=i+1,\dots,j-1$ such that
\begin{equation}
\label{N:vel.cdz}
|v_{k}^{+}\cdot v_{j}^{+}|\leq C \ep^{\nu/2},
\end{equation} 
namely $v_k^{+}$ is almost orthogonal to $v_j^{+}$ (see Figure \ref{N:fig:Recollision}). 
\begin{figure}[htbp] 
   \centering
  \includegraphics[scale=0.4]{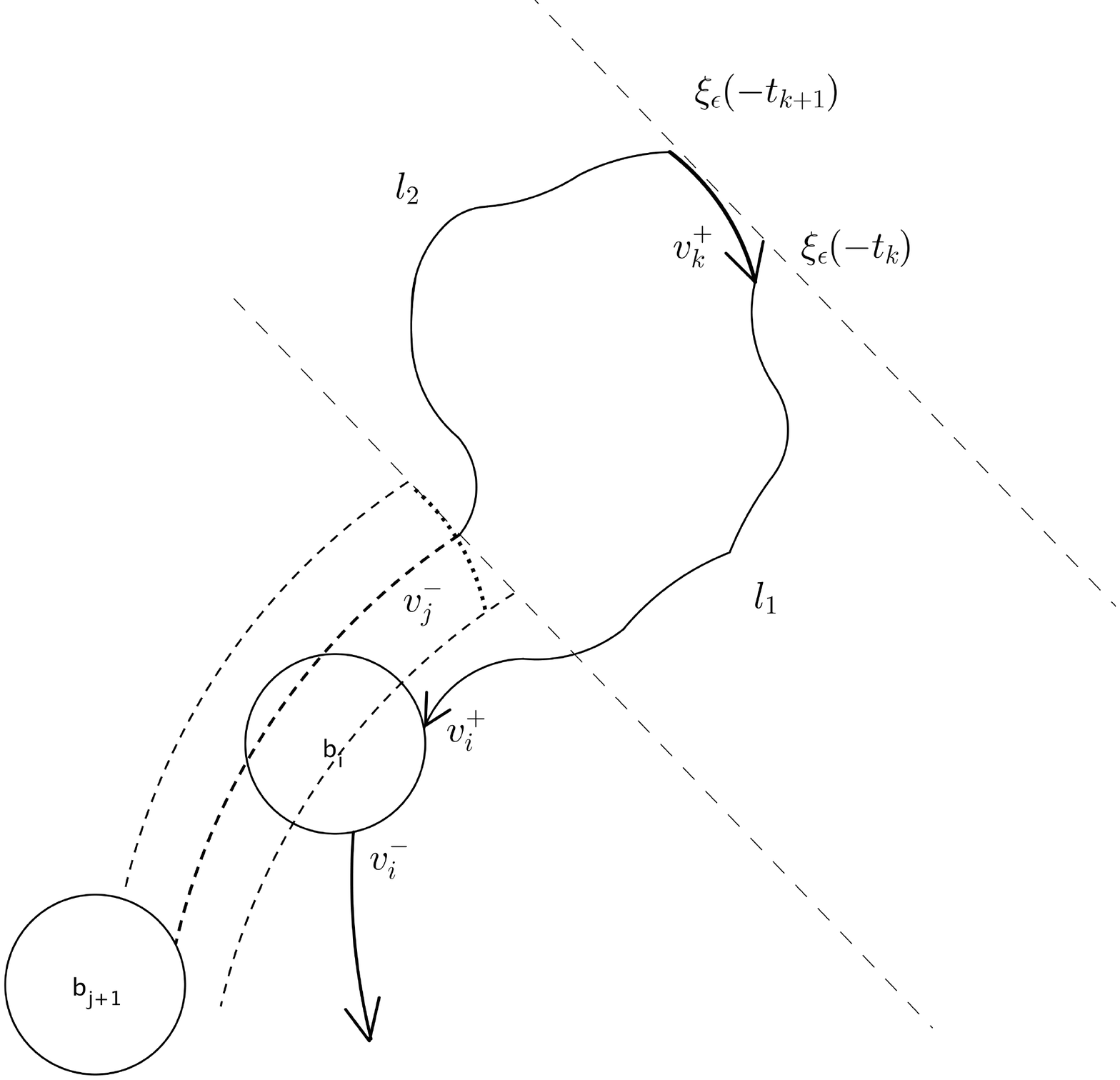} 
   \caption{Backward recollision.}
   \label{N:fig:Recollision}
\end{figure}

Then
\begin{equation} 
\chi_{rec}\big(1-\chi_{arc}^{(\nu)}\big)\leq \big(1-\chi_{arc}^{(\nu)}\big)\sum_{i=1}^{Q}\sum_{j=i+2}^{Q}\sum_{k=i+1}^{j-1}\chi_{rec}^{i,j,k},
\end{equation}
where $\chi_{rec}^{i,j,k}=1$ if and only if $\chi_{rec}^{i,j}=1$ and \eqref{N:vel.cdz} is fulfilled.
Following \cite{BNP}, we fix all the parameters $\rho_1,\dots,\rho_{Q}$, $t_1,\dots,t_Q$ but $t_{k+1}$. The two branches of the trajectory $l_1,l_2$ are rigid so that, when a recollision occurs,
the integration domain with respect to $t_{k+1}$ is restricted to a time interval bounded by 
$$\frac{2\ep}{\cos C\ep^{\nu/2}}\leq {4\ep}.$$ 	
	
Performing all the other integrations and summing over $i,j,k$ we obtain 
\begin{equation}
\begin{split}
\label{N:Markov error rec}
&\EE _{x,v}\left[\big(1-\chi_{arc}\big)\sum_{i=1}^{Q}\sum_{j=1}^{Q}\sum_{k=i+1}^{j-1}{\chi}_{rec}^{i,j,k}\right]\\&
\leq {C\ep}\, e^{-2\mu \ep^{-2\alpha}t} \sum_{Q\geq 3}(Q-1)(Q-2)(Q-3)\frac{(2 \mu \,\ep^{-2\alpha})^Q}{(Q-1)!} t^{Q-1}\\&
\leq C' t^3\ep^{1-8\alpha}, 
\end{split}
\end{equation}
which tends to 0 as $ \ep$ goes to 0 for $\alpha<1/8$.

Following the strategy used in \cite{BNP}, since a backward interference is a forward recollision, the estimate for the interference event can be handled by using the Liouville Theorem.  

\section{Proof of Proposition \ref{N:prop1LR}}\label{bounds_Boltzmann}
Our aim is to prove the asymptotic equivalence of $f_\ep$, defined by \eqref{N:valatt2}, and $\h$ solution of the linear Boltzmann equation \eqref{N:cutBoltz}, namely
\begin{equation*} 
\left\{\begin{array}{ll}
(\partial_{t}+v\cdot\nabla_{x}+B\,v^{\perp}\cdot\nabla_{v})\h (t,x,v)=\tilde{\text{L}}\h (t,x,v)&\\
\h(x,v,0)=f_0(x,v)&,
\end{array}\right.
\end{equation*}
which reads
\begin{equation}
\label{N:formula6}
\h(x,v,t)=e^{-2\mu\ep^{\gamma-1} t}\sum_{Q\geq 0}\mu_{\ep}^{Q}\int_{0}^{t}dt_Q\dots\int_{0}^{t_{2}}dt_1
\int_{-\ep^{\gamma}}^{\ep^{\gamma}}d\rho_1\dots\int_{-\ep^{\gamma}}^{\ep^{\gamma}}d\rho_Q\, f_0(\bar{\xi}_{\ep}(-t),\bar{\eta}_{\ep}(-t)).
\end{equation}
We recall that we can expand $ f_{\ep}$ as follows:
\begin{equation}
f_{\ep}(x,v,t)=e^{-\mu_{\ep}|\mathcal{B}(x,t)|}\sum_{N\geq 0}\frac{\mu_{\ep}^{N}}{N!}\int_{\mathcal{B}(x,t)^N}d\mbf{c}_{N}\, f_0(T^{-t}_{\mbf{c}_{N}}(x,v)),
\end{equation}
where $T^t_{\mbf{c}_{N}}(x,v)$ is the Hamiltonian flow with initial datum $(x,v)$ and $\mathcal{B}(x,t)$ is the disk of center $x$ and radius $t$.
We observe that the proof follows the same strategy of Proposition \ref{N:prop1LANDAU} (see also Section 3 in \cite{DP}). As before the hardest part is the estimate of the non-Markovian contribution which is summarized in the following proposition.
 \begin{proposition}\label{eq:pathologies_long_range}
Let $\ffi_1(\ep,t)$ be defined as in \eqref{N:ffi_1} with the only difference that the radius of the obstacles is now $ \ep^{\gamma}$ instead of $ \ep$ and the collision time $ \tau$ is bounded by $ C \ep^{\gamma}$ instead of $ C \ep$. Then for any $t\in[0,T]$
\begin{equation*}
\|\ffi_1(\ep,t)\|_{L^1}\to 0
\end{equation*}
as $\ep\to 0$.
\end{proposition}
\begin{proof}
Also in this case we can skip the estimates of the contributions $\EE _{x,v}[(1-\chi_1)]$ and $\EE _{x,v}[\chi_{ov}]$ since the presence of the external field does not affect the bounds in \cite{DP}. Moreover, as in proposition \ref{N:th:prop}, if we know that $\EE_{x,v}[\chi_{rec}]$ is negligible, then the Liouville theorem guarantees that also $\EE_{x,v}[\chi_{int}]$ can be disregarded in the limit. Hence, it suffices to focus on $\EE_{x,v}[\chi_{circ}]$ and $\EE_{x,v}[\chi_{rec}]$. 
So we look at
\begin{align}
\EE_{x,v}[\chi_{circ}] \leq c_{\gamma}(\ep)+ e^{-2t \mu_{\varepsilon}\varepsilon^{\gamma} }\sum_{Q\geq 1}(2 \mu_{\varepsilon}\varepsilon^{\gamma})^Q \int_0^t dt_1 \cdots \int_0^{t_{Q-1}}dt_Q \sum_{i=0}^{Q-1}\chi(t_{i+1}<t_i - T_L+C\varepsilon^{\gamma})
\end{align}
where $c_{\gamma}(\ep) = C \exp(-\frac{2 \pi \mu}{B} \ep^{\gamma-1}) $, $ \mu_{\varepsilon}=\mu\, \varepsilon^{-1}$ and $\gamma\in (0,1)$. Following the same strategy as in proposition \ref{N:th:prop}, we obtain 
\begin{align}
\EE_{x,v}[\chi_{circ}] & \leq  \, c_{\gamma}(\ep) + e^{-2t \mu_{\varepsilon}\varepsilon^{\gamma} }\sum_{Q\geq 1}(2\mu_{\varepsilon}\varepsilon^{\gamma})^Q \bigg[ \frac{2(t-T_L + C\varepsilon^{\gamma})^Q}{(Q-1)!} + \frac{(T_L-C\varepsilon^{\gamma})(t-T_L+C\varepsilon^{\gamma})^{Q-1}}{(Q-1)!}\bigg] \nonumber \\
& \, =\,c_{\gamma}(\ep) + 2e^{-2(T_L-C\varepsilon^{\gamma}) \mu \varepsilon^{\gamma-1} }(2 \mu \varepsilon^{\gamma-1})(t-T_L + C\varepsilon^{\gamma}) + e^{-2(T_L-C\varepsilon^{\gamma}) \mu \varepsilon^{\gamma-1}}(2 \mu\varepsilon^{\gamma-1}) (T_L-C\varepsilon^{\gamma})
\end{align}
which vanishes as $ \varepsilon \to 0$ and $\gamma\in (0,1)$. 

To control the recollision event we can follow the strategy used in Section \ref{N:pat} and in \cite{DP}, Proposition 3.1. More precisely, as in Section \ref{N:pat}, we introduce $\chi_{arc}^{(M)}$ such that $\chi_{arc}^{(M)}=1$ if there exists an entrance time $t_i$ such that $|t_{i}-t_{i+1}|\geq {T_L}/{M}-\tau\geq {T_L}/{M}-C\varepsilon^{\gamma}$ for some $i=0,\dots Q-1$ where $ M$ is a finite constant and $ M >1$.
One can easily see that $\EE_{x,v}[\chi_{arc}^{(M)}]$ vanishes as $ \varepsilon \to 0$ when $\gamma< 1$.  
Furthermore, 
$$\EE_{x,v}[\chi_{rec}]\leq \EE_{x,v}[\big(1-\chi_{arc}^{(M)}\big)\chi_{rec}]+\EE_{x,v}[\chi_{arc}^{(M)}]$$
but $\EE_{x,v}[\chi_{arc}^{(M)}]$ is vanishing in the limit $ \varepsilon \to 0$ as we have seen before. Therefore, we can focus on $\EE_{x,v}[\big(1-\chi_{arc}^{(M)}\big)\chi_{rec}]$.
We now distinguish the collisions as
\begin{equation}\label{eq:splitchirecLR}
\begin{split}
\big(1-\chi_{arc}^{(M)}\big)\chi_{rec}&\leq\,\big(1-\chi_{arc}^{(M)}\big) \sum_{i=1}^{Q}\sum_{j>1}\chi^{i,j}_{rec}\,\chi\big(\sin\alpha_{jk}\leq \frac{ \ep^{\delta}}{4},\;\forall k=i,\dots,j-1\big)\\
&\;+\big(1-\chi_{arc}^{(M)}\big) \sum_{i=1}^{Q}\sum_{j>1}\chi^{i,j}_{rec}\,\chi\big(\sin\alpha_{jk}\geq \frac{ \ep^{\delta}}{4},\; \text{ for some } k=i,\dots,j-1\big)
\end{split}
\end{equation}
where  $\chi_{rec}^{i,j}=1$ if and only if $b_i$ (constructed via the sequence $t_1,\rho_1,\dots,t_i,\rho_i$) is recollided in the time interval $(t_j,t_{j-1})$ and $\alpha_{jk}$ (with $ i<k<j$) is the absolute value of the sum of the angles between the outgoing velocity $v_k^{+}$ from the $ k$-th obstacle and the recolliding velocity ${v}_j^{-}$, i.e. 
$$
\alpha_{jk}= |\varphi_k| + \sum_{r=k+1}^{j} |\theta_r + \varphi_r| 
$$
where $ \theta_r$ is the deflection angle due to the $ r$-th scatterer and $ \varphi_r$ is the angle covered in the time interval $(t_{r+1},t^+_{r}) $ outside the scatterers and $ \varphi_j$ is the angle covered between the $ j$-th obstacle before recolliding with $ b_i$. 
Here $\delta>0$ is a suitable parameter that we will fix later.
Note that, thanks to $1-\chi_{arc}^{(M)}$, we have $|\varphi_r|\leq 2 \pi /{M}$ for any $ r$.
 
As noticed in \cite{DP}, the constraint $\sin\alpha_{jk}\leq \ep^{\delta}/ 4$ implies that 
$|\theta_r+\varphi_r-\pi|<\ep^{\delta}$ for some $r=i,\dots,j-1$, thus we get
\begin{align}\label{eq:ricLR1}
& \EE_{x,v}\big[\big(1-\chi_{arc}^{(M)}\big)\sum_{i=1}^{Q}\sum_{j>1}\chi^{ij}_{rec}\,\chi\big(\sin\alpha_{jk}\leq \frac{ \ep^{\delta}}{4},\;\forall k=i,\dots,j-1\big)\big] \nonumber \\ 
& \leq e^{-2t \mu_{\ep}\ep^{\gamma}}\sum_{Q \geq 0} \mu_{\ep}^Q \sum_{i=1}^Q \sum_{j=i+2}^Q \sum_{k=i}^{j-1} \int_0^t dt_1 \int_0^{t_1} dt_2 \cdots \int_0^{t_{Q-1}} dt_Q \nonumber \\
& \int_{-\ep^{\gamma}}^{\ep^{\gamma}}d\rho_1 \int_{-\ep^{\gamma}}^{\ep^{\gamma}}d\rho_2 \cdots \int_{-\ep^{\gamma}}^{\ep^{\gamma}}d\rho_Q \chi({|\theta_k+\varphi_k-\pi|<\ep^{\delta}}) .
\end{align}
Now we note that for $ M$ big enough
\begin{align}\label{ineq:rec}
\int_{-\ep^{\gamma}}^{\ep^{\gamma}} d\rho_k\, \chi({|\theta_k+\varphi_k-\pi|<\ep^{\delta}})= &\int_{-\pi}^{\pi} d \theta_k\check\Gamma ^{(B)}_{\ep,\gamma}(\theta_k)\chi({|\theta_k+\varphi_k-\pi|<\ep^{\delta}}) \nonumber \\
= & \, \ep \int_{-\pi}^{\pi} d \theta_k \Gamma ^{(B)}_{\ep,\gamma}(\theta_k)\chi({|\theta_k+\varphi_k-\pi|<\ep^{\delta}}) \nonumber \\
= & \, \ep \int_{\pi\left(1-\frac{2}{M}\right)-\ep^{\delta}}^{\pi\left(1-\frac{2}{M}\right)+\ep^{\delta}} d \theta_k \Gamma ^{(B)}_{\ep,\gamma}(\theta_k) \nonumber \\
\leq & \, C\ep^{1+\delta}
\end{align}
where $ \check\Gamma ^{(B)}_{\ep,\gamma}(\theta_k)$ is the differential cross section associated to the rescaled potential $ \check{\psi}_{\ep}$, while $ \Gamma ^{(B)}_{\ep,\gamma}(\theta_k)$ is differential cross section associated to the unrescaled potential $ {\psi}_{\ep}$. In the last line of \eqref{ineq:rec} we used that $\Gamma ^{(B)}_{\ep,\gamma}(\theta)$ is uniformly bounded in $\ep$ when $ \theta$ is far from 0, as shown in Appendix \ref{app:CS}.

Then from \eqref{eq:ricLR1} one gets
\begin{align}
& \EE_{x,v}\big[\big(1-\chi_{arc}^{(M)}\big)\sum_{i=1}^{Q}\sum_{j>1}\chi^{i,j}_{rec}\,\chi\big(\sin\alpha_{jk}\leq \frac{ \ep^{\delta}}{4},\; \forall k=i,\dots,j-1\big)\big] \nonumber \\ 
& \leq\, C e^{-2t \mu_{\ep}\ep^{\gamma}}\sum_{Q \geq 1} \frac{(2t \mu_{\ep}\ep^{\gamma})^Q}{Q!}Q^3\ep^{1+\delta -\gamma} \leq C \ep^{2 \gamma -2 +\delta} \,.
\end{align}

For what concerns the second term in \eqref{eq:splitchirecLR} we note that, once we fix all the variables $\{t_{\ell}\}_{\ell=1}^Q$ and $\{\rho_{\ell}\}_{\ell=1}^Q$ except $t_k$, by using the same geometrical argument as the one illustrated in Figure 5 of \cite{DP}, one gets that the integral over $t_{k+1}$ is bounded by 
\begin{align}
\frac{2 \ep^{\gamma}}{\sin\alpha_{jk}}\leq 8 \, \ep^{\gamma-\delta}\,.
\end{align}
It follows that 
\begin{align}\label{eq:ricLR2}
& \EE_{x,v}\big[\big(1-\chi_{arc}\big) \sum_{i=1}^{Q}\sum_{j>1}\chi^{ij}_{rec}\,\chi\big(\sin\alpha_{jk}\geq \frac{ \ep^{\delta}}{4},\; \text{ for some } k=i,\dots,j-1\big)\big] \nonumber \\
& \leq \, e^{-2t \mu_{\ep}\ep^{\gamma}} \sum_{Q \geq 0} \frac{(2  \mu_{\ep}\ep^{\gamma})^Q}{(Q-1)!} Q^3 t^{Q-1} 8 \,\ep^{\gamma - \delta} \nonumber \\
& \leq C T \ep^{5\gamma-\delta-4}.
\end{align}
We can now optimize the parameter $\delta$ setting $\delta=\frac{3\gamma-2}{2}$. From equations \eqref{eq:ricLR1} and \eqref{eq:ricLR2} we finally end up with
\begin{equation}
\EE_{x,v}\big[\chi_{rec}\big]\leq C \ep^{\frac{7\gamma-6}{2}} 
\end{equation}
which tends to 0 as $\ep \to 0 $ if $ \gamma \in (6/7,1)$.
\end{proof}

\bigskip

\textbf{Acknowledgments}\\
We thank Mario Pulvirenti for the proposed problem, Jens Marklof and Jani Lukkarinen for useful discussions on the topic. Finally we thank the anonymous referees for helpful suggestions to improve the manuscript. The research of M.\ Marcozzi and A.\ Nota has been supported by the 
Academy of Finland.
\vspace{10mm}
 
\appendix 
\section{The collision time}\label{app:CTime}
We want to estimate the time spent by a the test particle in the interaction disk associated to the central potential of finite range with a uniform magnetic field perpendicular to the plane. Let be $ \ep^{\alpha}\phi(r)$ with $ \alpha \in [0,1/2)$ the central potential and $ \ep B$ the modulus of the magnetic field. \\
The Lagrangian of the system is
\begin{equation*}
\mathcal{L}(r,\dot{r},\theta,\dot{\theta})=\frac{1}{2}\dot{r}^2+\frac{1}{2}r^2\dot{\theta}^2-\ep^{\alpha}\phi( r)+\varepsilon\frac{B}{2}r^2\dot{\theta}
\end{equation*}
We observe that the energy of the system is conserved. Moreover the Lagrangian does not depend on the variable $\theta$, so we obtain the conservation of the conjugate momentum
\begin{equation*}
\frac{\,d}{\,dt}(r^2\dot{\theta}+\varepsilon\frac{B}{2}r^2)=0.
\end{equation*}
Therefore we obtain the following conserved quantities 
\begin{equation*}
\dot{r}^2+r^2\dot{\theta}^2+2\ep^{\alpha}\phi=2E,
\end{equation*}
\begin{equation*}
r^2\dot{\theta}+\varepsilon\frac{B}{2}r^2=M,
\end{equation*}
and the equations of motion are
\begin{equation*}
\left \{
\begin{array}{l}\vspace{2mm}
\dot{r}:=\frac{\,dr}{\,dt}=\sqrt{2(E-\ep^{\alpha}\phi)-\frac{M^2}{r^2}-\frac{\varepsilon^2 B^2}{4}r^2+\varepsilon MB}\\ 
\dot{\theta}:=\frac{\,d\theta}{\,dt}=\frac{M}{r^2}-\varepsilon\frac{B}{2}.\\
\end{array}
\right.
\end{equation*}
This implies that
\begin{align}\label{integral_equations}
\frac{dt}{dr}&=\bigg[{2(E-\ep^{\alpha}\phi(r))-\frac{M^2}{r^2}-\frac{\varepsilon^2 B^2}{4}r^2+\varepsilon MB}\bigg]^{-1/2} \nonumber \\
\frac{d\theta}{dr}&= \frac{\frac{M}{r^2}-\frac{\varepsilon B}{2}}{\sqrt{2(E-\ep^{\alpha}\phi(r))-\frac{M^2}{r^2}-\frac{\varepsilon^2 B^2}{4}r^2+\varepsilon MB}} .
\end{align}
We now define the effective potential
\begin{equation*}
\phi_{eff}( r)=\ep^{\alpha}\phi( r)+\frac{M^2}{2r^2}+\frac{\varepsilon^2 B^2}{8}r^2-\frac{\varepsilon MB}{2}.
\end{equation*}
and we assume that the potential has a short range, i.e. $ \phi(r):[0,1] \to \R$ and $ \phi$ is continuous on $ [0,1]$ and differentiable on $(0,1) $.

Take the modulus of the initial velocity to be $ | v|=1$. When the particle hits the obstacle of radius $r=1$ the conserved quantities read
\begin{equation*}
\left \{
\begin{array}{l}
E=\frac{1}{2}\\
M=\rho+\varepsilon\frac{B}{2}\\
\end{array}
\right.
\end{equation*}
being $\rho \in [0,1]$ is the impact parameter. The effective potential is
\begin{equation*}
\phi_{eff}( r)=\ep^{\alpha}\phi( r)+\left(\rho+\varepsilon\frac{B}{2}\right)^2\frac{1}{2r^2}+\frac{\varepsilon^2 B^2}{8}r^2-\left(\rho+\varepsilon \frac{B}{2}\right)\frac{\varepsilon B}{2}=\ep^{\alpha}\phi( r) + \frac{1}{2}\bigg[ \frac{\rho}{r}-\frac{\varepsilon B}{2}\bigg(r-\frac{1}{r} \bigg)  \bigg]^2.
\end{equation*}
By integrating the equations of motion we obtain the collision time, namely the time spent inside the obstacle:
\begin{equation}\label{coll_time}
\tau=2\int_{r_{min}}^{1}dr \bigg[1-2\ep^{\alpha}\phi( r)-\bigg( \frac{\rho}{r}-\frac{\varepsilon B}{2}\bigg(r-\frac{1}{r} \bigg)  \bigg)^2 \bigg]^{-1/2}
\end{equation}
where $r_{min}$ (the minimum distance from the centre) 
 is the unique zero of the radicand, i.e.
\begin{equation*}
1=2\phi_{eff}(r_{min} ),
\end{equation*}
so we can reformulate (\ref{coll_time}) as
\begin{equation*}
\tau=\sqrt{2}\int_{r_{min}}^{1}\frac{dr}{\sqrt{2(\phi_{eff}(r_{min} )-\phi_{eff}(r ))}}
\end{equation*}
where $2\phi_{eff}(r )\leq 1$. 
The derivative of the effective potential reads 
\begin{equation*}
\phi'_{eff}( r)=\ep^{\alpha}\phi'( r)-\frac{\rho^2}{r^3}-\frac{\varepsilon^2 B^2}{4r^3}-\frac{\varepsilon B\rho}{r^3}+\frac{\varepsilon^2 B^2}{4}r.
\end{equation*}
By the mean value theorem we get 
\begin{equation*}
\vert \phi_{eff}(r_{min} )-\phi_{eff}(r )\vert =\vert r-r_{min}\vert \vert -\phi'_{eff}(r^{*}) \vert \geq \vert r-r_{min}\vert \bigg(\inf_{r \in (r_{min},1)}\vert -\phi'_{eff}(r) \vert \bigg) , \qquad r^{*}\in (r_{min},r)
\end{equation*}
and then
\begin{equation*}
\tau \leq\frac{\sqrt{2}}{(\inf_{r \in (r_{min},1)}\vert -\phi'_{eff}( r)\vert)^{1/2}}\int_{r_{min}}^{1}{\frac{1}{\sqrt{r-r_{min}}}\,dr}.
\end{equation*}
Since $\phi'_{eff}( r) < 0$ for $ r \in [0,1)$, then $ \inf_{r \in (r_{min},1)}\vert -\phi'_{eff}( r)\vert =: \kappa >0$ and it follows easily that
\begin{equation*}
\tau \leq 2\bigg(\frac{2(1-r_{min})}{\kappa}  \bigg)^{1/2} \leq 2\bigg(\frac{2}{\kappa}  \bigg)^{1/2}.
\end{equation*}                                        
For the corresponding rescaled problem the effective potential reads
\begin{align}
\phi_{eff}^{(\varepsilon)}(r)= \varepsilon^{\alpha}\phi(r/\varepsilon)+\frac{1}{2}\bigg[\frac{\varepsilon\rho}{r}+\frac{B}{2}\bigg(\frac{\varepsilon^2}{r}-r  \bigg)   \bigg]^2
\end{align} 
with $ \rho \in [0,1)$. In this way one gets $
- \phi_{eff}^{(\varepsilon)'}(r)=\frac{1}{\varepsilon}F(r/\varepsilon, \varepsilon) $ where
\begin{align}
F(y, \varepsilon)=-\varepsilon^{\alpha}\phi'(y)+(\rho +  B)\rho y^{-3}+\frac{\varepsilon^2 B^2}{4 y^3}(1-y^4) 
\end{align}
which is positive for $ y < 1$ and uniformly in $ \varepsilon$. 
The same argument as before yields the claimed estimate: 
\begin{align}
\tau_{\varepsilon} \leq \frac{(2 \varepsilon)^{1/2}}{(\inf_{y \in (y_0,\varepsilon)} F( y,\varepsilon))^{1/2}}\int_{y_0}^{\varepsilon}{\frac{dy}{\sqrt{y-y_{0}}}} \leq C \varepsilon
\end{align}
where $ y_0=y_0(\varepsilon)$ is such that $$ 1=2 \varepsilon^{\alpha}\phi(y_0)+ \bigg(\frac{\rho}{y_0}-\frac{\varepsilon^2 By_0}{2} + \frac{B}{2y_0}  \bigg)\,.$$ 
We consider now the long range unrescaled potential defined in Assumption $B1)$, i.e. $\psi_{\ep}( r)=r^{-s}$ for $ r < \varepsilon^{\gamma-1}$ and $\psi_{\ep} (r)=\varepsilon^{-s(\gamma-1)}$ for $ r \geq \varepsilon^{\gamma-1}$.                                   The same argument as for the short range case leads to the following estimate for the collision time after rescaling:
\begin{align}
\tau \leq \frac{(2 \varepsilon)^{1/2}}{(\inf_{y \in (y_0,\varepsilon^{\gamma})} \tilde{F}( y,\varepsilon))^{1/2}}\int_{y_0}^{\varepsilon^{\gamma}}{\frac{dy}{\sqrt{y-y_{0}}}} \leq C \varepsilon^{\gamma}
\end{align} 
with $$\tilde{F}(y,\varepsilon):=-\psi_{\ep}'(y)+\varepsilon^{2(\gamma-1)}\rho^2 y^{-3} + B^2 \varepsilon^{3 \gamma-2}\rho y^{-3} + B^2 y^{-3}\varepsilon^{4 \gamma-2}(1-y^4 \varepsilon^{4(1-\gamma)})/4 > 0$$
for $ y \in (y_0, \varepsilon^{\gamma})$ where $y_0=y_0(\varepsilon) $ is such that $1= 2\psi_{\ep}(y_0)+ \bigg[ \varepsilon^{\gamma-1}\rho y_0^{-1}-\frac{B}{2}(\varepsilon^{2\gamma-1}y_0^{-1}-\varepsilon y_0) \bigg]^2 $.

\section{Cross section} \label{app:CS}
\begin{proposition}\label{prop:thetaCR}
Consider the scattering angle $\theta(\rho, \ep)$ of a particle
with impact parameter $\rho$ due to a uniform magnetic field perpendicular to the plane with modulus $ \ep B$ and due to a radial potential $ \ep^{\alpha}\phi$, where $\alpha > 0$ and $ \phi$ satisfies assumptions A1, A2, A3. Consider also the scattering angle $ \tilde\theta(\rho, \ep)$ associated to the same radial potential as before, but without any magnetic field. Then, for $ \ep$ small enough one gets 
\begin{align}
\theta(\rho, \ep) = \tilde\theta(\rho, \ep) + O(\ep)\,.
\end{align}
\end{proposition}
\begin{proof}
Following \cite{DR, L} we can write the exact formula for both of the scattering angles:
\begin{align}
\tilde\theta(\rho, \ep) = \pi - 2\arcsin \rho- 2 \int^{\tilde{u}_{max}(\rho, \ep)}_{\rho} \frac{du}{\sqrt{1- u^2 - 2\ep^{\alpha}\phi(\rho u^{-1})}}
\end{align}
where $ \tilde{u}_{max}(\rho, \ep)$ is the solution of the equation $ \tilde{u}^2_{max}+ 2\ep^{\alpha}\phi(\rho \tilde{u}_{max}^{-1})=1$, while
\begin{align}
\theta(\rho, \ep) = \pi - 2\arcsin \rho -2 \int^{u_{max}(\rho, \ep)}_{\rho}du \frac{1 + \frac{\ep B}{2 \rho}(1 - \frac{\rho^2}{u^2})}{\sqrt{1- 2\ep^{\alpha}\phi(\rho u^{-1})-u^2[1+\frac{\ep B}{2 \rho}(1-\frac{\rho^2}{u^2})]^2}} 
\end{align}
where $ {u}_{max}(\rho, \ep)$ is the solution of the equation $ 2\ep^{\alpha}\phi(\rho u_{max}^{-1})+u_{max}^2[1+\frac{\ep B}{2 \rho}(1-\frac{\rho^2}{u_{max}^2})]^2=1$. 

Hence, an expansion of $ \theta(\rho, \ep)$ for $ \ep$ small enough yields the claimed asymptotic formula. 

\end{proof}

\begin{proposition}\label{RK:c-s}
Let $ \tilde\theta$ be the scattering angle associated to the long range potential $ \Psi(r)=r^{-s}$ with $ s>2$, $ \theta_{\ep, \gamma}$ the scattering angle due to a radial potential $ \psi_{\ep}$ defined in Assumption B1) and $ \theta_{\ep, \gamma}^{(B)}$ the scattering angle due to $ \psi_{\ep}$ and to a uniform, constant magnetic field perpendicular to the plane with modulus $ \ep B$. Then one has
\begin{itemize}
\item[a)] 
$ \theta_{\ep, \gamma}^{(B)} \to \tilde\theta$ as $\ep \to 0$.
\item[b)] 
$ \Gamma^{(B)}_{\ep, \gamma}(\theta) \to \Gamma(\theta)$ as $ \ep \to 0$, where $ \Gamma^{(B)}_{\ep, \gamma}(\theta)$ is the differential cross section associated to the radial potential $ \psi_{\ep}$ and the magnetic field, while $\Gamma(\theta)$ is the one associated to the radial potential $ \Psi$. 
\item[c)] $ \Gamma^{(B)}_{\ep, \gamma}(\theta) \leq C \theta^{-1-1/s}$ uniformly in $ \ep,B$.
\end{itemize} 
\end{proposition}

\begin{proof} a) Let us now consider the truncated potential $ \tilde{\Psi}=r^{-s}-A^{-s}$ with $ s>2$ for $ r \leq A$ and $ \tilde{\Psi}=0$ for $ r> A$ with $ A=\varepsilon^{\gamma-1}$ and $ \gamma \in (0,1)$. Take the modulus of the initial velocity of the light particle to be $ |v|=1$.

We denote by $\rho$ the impact parameter (with $0\leq\rho\leq A$) while the scattering angle (that is the angle between the ingoing and the outgoing relative velocities) is
\begin{equation}\label{theta_rho}
{\theta}_{\ep,\gamma}(\rho)=2\int_{\arcsin(\rho/A)}^{\pi /2}\bigg(1- \frac{\sin \beta}{v + s v^{s-1}}   \bigg)d\beta.
\end{equation}
where $ v=v(\beta)$ such that $v^2 + 2((v/\rho)^s - A^{-s}) = \sin^2 \beta $ and $ v=\rho / r$ (see Appendix in \cite{DP}). 
Following \cite{DP,L}, we can write the formula for the scattering angle associated to the potential $ \tilde{\Psi}$ and with the uniform magnetic field. Due to its invariance under rescaling, the scattering angle associated to the equations of motion \ref{N:scaled2} reads
\begin{align}\label{thetaM}
\theta^{(B)}_{\ep,\gamma}(M)=& \pi -2\arcsin\left(\frac{M}{A}-\ep \frac{B A^2}{2}\right)-2\int_{r_*}^{A}\frac{\frac{1}{r^2}(M-\ep\frac{B}{2}r^2)\,dr}{\sqrt{1-2\tilde{\Psi}_{eff}(r)}} \nonumber \\
= & 2 \int_{\arcsin (\frac{M}{A}-\frac{\ep BA}{2})}^{\pi /2} d \beta \bigg[1-\frac{(1-\frac{\ep BM}{2 u^2})\sin \beta}{u + su^{s-1}M^{-s} + \frac{\ep^2 B^2 M^2}{4 u^3}} \bigg]
\end{align}
where $\tilde{\Psi}_{eff}(r)=\tilde{\Psi}(r)+\frac{1}{2}\left(\frac{M}{r}-\frac{\ep}{2}Br\right)^2$,
$M$ is the value of the conserved momentum at the hitting time, i.e. $M=\rho+\ep A^2\frac{B}{2}$,
and $r_*$ is defined as the solution of the equation $2\tilde{\Psi}_{eff}(r_*)=1$. In the second line we made the change $ r \to u \to \beta$ where $ u=u(\beta,M)=\frac{M}{r}$ and $ \sin^2 \beta= 2 \tilde{\Psi}(M/u)$. Note that the change of variable $ u \to \beta$ is well defined because $ \tilde{\Psi}(M/u)$ is non-decreasing when $ u \in [M/A,M/r^*]$ for $ \ep$ small enough.

From \eqref{theta_rho} and \eqref{thetaM} it is clear that $ {\theta}_{\ep,\gamma}$ and $ \theta_{\ep,\gamma}^{(B)}$ have the same asymptotic behaviour as $ \ep$ approaches 0. 
Since ${\theta}_{\ep,\gamma}\to \tilde\theta$, one gets the claim.

b) The inverse of the differential cross section associated to $ \tilde\theta$ is
\begin{align}\label{eq:inv_cross_no_B}
\bigg\vert \frac{d \theta_{\ep,\gamma}}{d \rho}\bigg \vert = \frac{2}{\rho^{s+1}}\int_{\arcsin(\rho /A)}^{\pi /2} \frac{d \beta \sin \beta s v^{s-1}}{(v + sv^{s-1}\rho^{-s})^2} \bigg[s-\frac{1 + s(s-1) v^{s-2}\rho^{-s}}{1+s v^{s-2} \rho^{-s}}  \bigg] + \frac{2 s \rho^{-2} A^{2-s}}{1+2 s \rho^{-2} A^{2-s}} \,.
\end{align}

We want to study the limit of $ d\theta_{\ep,\gamma}^{(B)}/d\rho$. For a mere computational convenience, we prefer to look at $ d\theta_{\ep,\gamma}^{(B)}/dM$ which is related to $ d\theta_{\ep,\gamma}^{(B)}/d\rho$ via 
\begin{align}\label{theta_M-theta_rho}
 \frac{d\theta_{\ep,\gamma}^{(B)}}{dM}(M)=\frac{d\theta_{\ep,\gamma}^{(B)}}{d\rho}(\rho + \frac{\ep BA^2}{2})\,.
\end{align}
From (\ref{thetaM}) one gets
\begin{align}\label{dtheta_M}
& \frac{d \theta_{\ep,\gamma}^{(B)}}{d M}= -\frac{2}{A \sqrt{1-(\frac{M}{A}-\frac{\ep BA}{2})^2}} \bigg[ \frac{s M^{-2}A^{2-s} + \ep B A^2 M^{-1}}{1 + sM^{-2}A^{2-s} + \frac{\ep^2}{4} B^2 A^4 M^{-2}} \bigg] \nonumber \\
&- 2 \int_{\arcsin (\frac{M}{A}-\frac{\ep BA}{2})}^{\pi /2}\frac{d \beta \sin \beta}{(u + su^{s-1}M^{-s} + \frac{\ep^2 B^2 M^2}{4 u^3})^2}\big[ s^2 u^{s-1}M^{-s-1}-u'(1+s(s-1)u^{s-2}M^{-s})   \big] \nonumber \\
& - {\ep B M} \int_{\arcsin (\frac{M}{A}-\frac{\ep BA}{2})}^{\pi /2}\frac{d \beta \sin \beta} {u}\bigg[s(s+1)u^{s-2}M^{-s}(u'M-1) + 3 u'M  \nonumber \\
& + \frac{5 \ep^2 B^2 M^3 u'}{4 u} + \frac{\ep BM}{u^2} + \frac{3 \ep B M^2}{2 u^3} -1 - \frac{3 \ep^2 B^2 M^2}{4 u^4} \bigg]
\end{align}
where
\begin{align}
u' = \frac{du}{dM} = \frac{1}{M^{s+1}}\bigg(\frac{su^{s-1} + (1-\frac{\ep BM}{2 u^2})\frac{\ep B M^{s+1}}{2u}}{1+su^{s-2}M^{-s} - \frac{\ep^2 B^2 M^2}{4 u^4}}\bigg) \,.
\end{align}
As for item a), one realizes that $ d\theta_{\ep,\gamma}^{(B)}/d\rho$ and $ d\theta_{\ep,\gamma}/d\rho$ are asymptotically equivalent for any $\rho$, thus Proposition A.1 in \cite{DP} implies that $\Gamma_{\ep, \gamma}^{(B)}(\theta) \to \Gamma(\theta)$ for $ \theta \in (-\pi, \pi)$ because its inverse map converges everywhere.

c) From (\ref{dtheta_M}) for $ \ep$ small enough a tedious expansion gives 
\begin{align}
\bigg\vert \frac{d \theta_{\ep,\gamma}^{(B)}}{d M} \bigg\vert \geq \bigg\vert \frac{d \theta_{\ep,\gamma}}{d \rho}\bigg\vert - \ep \vert \pazocal{R}(B,M,\ep)\vert \geq \frac{1}{C}\bigg\vert \frac{d {\theta}_{\ep,\gamma}}{d \rho}\bigg\vert
\end{align}
where $C >1 $ is a constant, $ \pazocal{R}$ is bounded in $ \ep$. The claim follows thanks to Proposition A.1 in \cite{DP}.
\end{proof}

\end{document}